\documentclass{article}
\usepackage{booktabs,amsthm}

\usepackage{amssymb}
\usepackage{relsize}
\usepackage{framed,amsmath}
\usepackage[]{algorithm2e}
\usepackage{tikz}
\usetikzlibrary{shapes.geometric,positioning}
\usepackage{pgfplots}
\usepackage{amsthm}
\newtheorem{definition}{Definition}
\newtheorem{claim}{Claim}
\newtheorem{note}{Note}
\newtheorem{lemma}{Lemma}
\newtheorem{corollary}{Corollary}
\newtheorem{example}{Example}
\newtheorem{theorem}{Theorem}
\usepackage{tikz}
\usepackage{subfig}
\usepackage{pgfplots}
\usepackage{manyfoot}
\newfootnote{B}
\makeatletter
\def\sfootnote#1{\ifx\protect\@typeset@protect
    \Footnotemark{*}\FootnotetextB{*}{#1}%
  \else
    \relax
  \fi
}

\usetikzlibrary{arrows}
\usetikzlibrary{shapes.geometric,positioning}

\tikzset{
  treenode/.style = {align=center, inner sep=0pt, text centered,
    font=\sffamily},
  arn_n/.style = {treenode, circle, white, font=\sffamily\bfseries, draw=black,
    fill=black, text width=1.5em},
  arn_r/.style = {treenode, circle, black, draw=black, 
    text width=1.5em, very thick},
  arn_x/.style = {treenode, rectangle, draw=black,
    minimum width=0.5em, minimum height=0.5em}
}

\begin{document}

\title{Attacking Power Indices by Manipulating Player Reliability\\
(extended abstract)}  

\author{Gabriel Istrate\thanks{corresponding author. Email: gabrielistrate@acm.org}, Cosmin Bonchi\c{s}, Alin Br\^{i}ndu\c{s}escu \\
Dept. of Computer Science, West University of Timi\c{s}oara, \\
Bd. V. P\^{a}rvan 4, 
Timi\c{s}oara, 300223, Romania.}

\maketitle
\begin{abstract}  
We investigate the manipulation of power indices in TU-cooperative games by stimulating (subject to a budget constraint) changes in the propensity of other players to participate to the game. 

We display several algorithms that show that the problem is often tractable for so-called network centrality games and influence attribution games, as well as an example when optimal manipulation is intractable, even though computing power indices is feasible. 
\end{abstract}


\section{Introduction}

\textit{Control} is a fundamental but difficult issue in multi-agent systems. A multi-agent society may be difficult to control due to the concurrence of several factors, that may interact and drive the dynamics in complex, unpredictable ways. Some of these factors could include uncertainty about agent involvement \cite{reliability-games},  coalition formation \cite{chalkiadakis2007coalition},  the rules  \cite{elkind2012manipulation}, the environment \cite{yokoo2005coalitional}, about rewards \cite{ketchpel1994forming}, the presence (or lack) of synergies between players \cite{procaccia2014structure}, etc.
 
 A common type of control is \textit{manipulation}\footnote{We  use the word with its wider, commonsense  meaning, rather than the specialized one from voting theory \cite{barriers-manipulation-voting}. Our usage encompasses both strategic behavior by an agent or coalition (voting theory "manipulation") and
interventions by a chair or outside agent (such as control and bribery in voting \cite{control-bribery-voting}). We assume, however, that all such interventions are costly.}, which often aims to change the power (index)  of a given player by means of interventions in the settings or the dynamics of the agent society. Many types of manipulation have been considered in the literature, often in a computational social choice context. They include  \textit{identity} \cite{aziz2011false}, \textit{cloning} \cite{elkind2011cloning} and \textit{quota}  \cite{zuckerman2012manipulating} manipulation in voting games, \textit{collusion and mergers} \cite{lev2013mergers}, \textit{sybil attacks} \cite{vallee2014study}   and, finally, \textit{multi-mode attacks} \cite{faliszewski2011multimode}, just to name a few. 

We contribute to this direction by studying yet another natural mechanism for  manipulation: \textbf{changing the propensity of players to participate to the game.} This type of manipulation is quite frequent in real-life situations, a central example being voting - while parties cannot control with absolute certainty voter turnout on election day, they may employ tactics that aim to mobilize their supporters and deter participation of their opponents' voters\footnote{Such scenarios are best modeled as \textit{multichoice voting games} \cite{felsenthal1997ternary}. However, since such games are  \textit{multi-cooperative}  (rather than cooperative) games \cite{bilbao2000bicooperative}, they fall outside of the scope of the present work, and will be dealt with in a subsequent paper.}. Manipulation could be performed by a centralized actor (like in the voting example), or 
by a coalition of players \cite{conitzer2007elections}, strategically modifying their behavior (in our case their reliabilities) in response to a perceived dominance of a player whose power index they wish to decrease.   

The main impetus for our work was \cite{bachrach2014cooperative}, where a model of strategic manipulation of player reliabilities was first investigated. Bachrach et al.  \cite{bachrach2014cooperative} considered  \textit{max games}. In these games each player possesses a weight; the value of a coalition is the maximum weight of a component of the coalition. They proved a "no sabotage theorem" for (the reliability extension of) max-games with a common failure probability. They remarked that manipulating player reliabilities can be studied in principle for all coalitional games, and asked for further investigations of this problem, in settings similar to the one we consider, i.e. under costly player manipulation.  Given the negative results for max-games \cite{bachrach2014cooperative} and the fact that computing power indices is often intractable  \cite{deng1994complexity}, we concentrate mostly on proving \textit{positive results}, showing that there exist natural scenarios where optimal attacks on power indices by manipulating players' reliabilities are easy to compute (and interesting).  We hope that these positive results will encourage renewed interest (and research) on the scope and limits of reliability manipulation. 

\textbf{Contributions and outline} In Section~\ref{stwo} we begin by informally stating the problem and justifying our choice of the two classes of coalitional games studied in this paper: \textit{network centrality games} \cite{suri2008determining,aadithya2010game,michalak2013efficient,tarkowski2017game} and \textit{credit attribution games} \cite{papapetrou2011shapley,karpov2014equal}.  Even though credit attribution games may seem to be somewhat exotic/of limited use, their importance extends well-beyond  scientometry: they were, in fact, anticipated, as \textit{hypergraph games} (see \cite{deng1994complexity} Section 3). The two games we consider from this class, \textit{full credit} and \textit{full obligation} games, are natural examples of read-once marginal contribution (MC) nets \cite{elkind2009tractable}. Full credit games are equivalent to the subclass of basic MC-nets \cite{ieong2005marginal} whose rules are \textit{conjunctions of  positive variables}; full obligation games correspond to generalized MC-nets whose rules consist of \textit{disjunctions of positive variables}. Full obligation games can simulate \textit{induced subgraph games} \cite{deng1994complexity}. Full credit games capture an important subclass of coalitional skill games (CSG) \cite{bachrach2008coalitional,bachrach2013computing}, that of CSG games with tasks consisting of a single skill. 

 Section~\ref{sthree} contains technical details and precise specifications of the models we investigate. We deal with two types of attacks: (node) \textit{removal}, where we are allowed to remove (decrease to zero the reliability of) certain nodes, and \textit{fractional attacks}, where reliability probabilities can be altered continuously. 
 
 In Section~\ref{sfive} we give closed-form formulas for the Shapley values of the reliability extensions of network centrality and credit allocation games.  Next we particularize these results to centrality games on specific network: first we show that no removal attack is beneficial; as for fractional attacks, we show that in the complete graph $K_{n}$ or when attacking the center of the star graph $S_{n}$, a greedy approach works: one should increase the reliabilities of  neighbors 
of the attacked node,  in descending order of baseline reliabilities. When attacking a non-center player  in $S_n$ the result is similar, with the important exception that increasing the reliability of the  center should precede all other moves. 
 In contrast, the situation for the cycle graph $C_{n}$ is more complex, involving all distance-two neighbors of the attacked node. A simple characterization is provided for the optimum as \textit{the best of four fixed ``greedy'' solutions}. This characterization allows the determination of the optimum for all combinations of reliability values and budget.\footnote{The precise formula for the optimum is cumbersome, hence deferred to the full version.}   
An interesting, and unintuitive, qualitative feature of the result is that in the optimal attack \textit{a non-neighbor of the attacked node could be targeted \textbf{before} some of the direct neighbors of the attacked node}.

 	In Section~\ref{ssix} we analyze full credit and full obligation games. Although these two games have the same Shapley value \cite{karpov2014equal}, we show that \textit{they behave very differently with respect to attacks}: removal attacks are not beneficial for full credit games, NP-hard for full obligation games. Fractional attacks also behave differently, modifying  probabilities in opposite directions. In a particular setting which includes the case of induced subgraph games we obtain greedy algorithms for both games, derived from expressing the problems as fractional knapsack problems. The determining quantities for the attack orders are (two different) cost-benefit measures.

\section{Problem Statement and Choice of Games}
\label{stwo}

The \textit{power index attack problem}, the main problem of interest in this paper, has the following simple informal statement: we consider the reliability extension of a cooperative game. We are given a positive budget $B>0$ and are allowed to modify reliabilities of all nodes, other than the targeted player $x$, as long as the total cost incurred  is at most $B$. Which nodes should we target, and how should we change their reliabilities, in order to decrease as much as possible the Shapley value of node $x$? 

A variant of the previous problem, called the \textit{pairwise power index attack problem} and motivated by Example~\ref{ex:hirsch} below, is the following: we are given not one but \textit{two} players $x,y$. The goal is to decrease as much (within the budget) the Shapley value of $x$,  while not affecting at all the Shapley value of $y$. This restriction makes some nodes exempt from attacks: we are not allowed to change the  reliabilities of players who contribute to the Shapley value of  $y$. 

\textbf{Choice of games} The problems described above could  be investigated in all classes of TU-cooperative games, or compact representation frameworks. However, we feel that the most compelling cases are those where the computation of power indices, e.g. the Shapley value,  of (the reliability extensions) of our games is tractable\footnote{This requirement disqualifies many natural candidate games such as \textit{weighted voting games} \cite{deng1994complexity,matsui2001np}, as well as most subclasses of \textit{coalitional skill games} \cite{aziz2009power}}. In other words \textit{the intractability of manipulating a power index should \textit{not} be a consequence of our inability to compute these indices}. In particular, we are interested in scenarios where computing power indices is easy, but computing an optimal attack on them is hard. Theorem~\ref{full-obligation} below provides such an example. 

The appeal of studying attacks on node centrality in social networks is quite self-evident: game-theoretic 
concepts such as those considered in \cite{suri2008determining,aadithya2010game,michalak2013efficient,tarkowski2017game} formalize appealing notions of leadership in social situations. They have been proposed as tools for identifying key actors, with applications e.g. to terrorist networks \cite{michalak2015defeating,lindelauf2013cooperative}. In such a setting, a direct (physical) attack on a leading node may be infeasible. Instead, one could attempt to indirectly affect its status (centrality), by incentivizing some of its peers.

Relevant examples of targeting nodes  in order to affect power indices arose (implicitly) in even earlier work  \cite{papapetrou2011shapley}, that attempted to develop coalitional  models of credit allocation in scientific work. The following is a version of the example in \cite{papapetrou2011shapley}:

\begin{example} \textit{Two scientists $A,B$ are compared with respect to their publication record\footnote{We \textbf{do not} condone and caution against the real-life use of such crude quantitative metrics for tasks like the one described in this example or our models.}. All their papers have exactly one co-author. Figure~\ref{hirsch} displays this information as a graph, listing for each author pair, the number of publications they have authored and the number of citations.  If using the Hirsch index, it would seem that candidate $A$ has a better track record than candidate $B$. But if we discard publications both of them have co-written with ``famous scientist $Y$'' (that is, \textbf{remove $Y$ and its publications from consideration}), then their relative ranking would be reversed.}

The authors of \cite{papapetrou2011shapley} attempted to use the Shapley value of a game based on the Hirsch index for credit allocation. An ulterior, more general and cleaner game-theoretic approach is ~\cite{karpov2014equal}. The author defines several  \textit{credit allocation games}, and uses their (identical) Shapley values as a measure of individual publication record. Slightly modified versions of this measure have (regrettably) actually been used in some countries to set minimum publication thresholds for access and promotion to academic positions, e.g. the minimal standards in Romania. 

In such a context one could naturally ask  the following question: \textit{what are the top $k$ coauthors that account for most of a scientists' publication record?} When using the game-theoretic framework for scientific credit from \cite{karpov2014equal}, this is equivalent to finding the $k$ coauthors whose removal (together with the joint papers) causes the scientist's' Shapley value to decrease the most.

Collaborations may, however, be genuinely productive or just bring to one of the scientists the benefits of association with leading scientists\footnote{One could argue, of course, that such an association itself reflects positively on the scientist. But the opposite argument, that prestige drives scientific inequality, has recently been substantiated by real data \cite{Morgan2018} and is, at the very least, hard to ignore.}. The Shapley value approach of \cite{karpov2014equal} does not distinguish between these two scenarios, as it gives equal credit to all authors of a joint paper, irrespective of "leadership status". Recent work, e.g. Hirsch's \textit{alpha index} \cite{hirsch2018mathbf}, has attempted to quantify "scientific leadership". It is possible to define a measure based on the reliability extension of credit allocation games that factors out the "well connectedness" of an individual from its score\footnote{The measure computes appropriate values of reliability probabilities, the lower the probability the more of a "scientific leader" a coauthor is; we are currently investigating the practicality of such an approach.}. Given such a measure, the previous question, that of finding the top-$k$ co-authors is still interesting, as it  \textit{identifies the most (genuinely) fruitful collaborations of a given author, irrespective of status.} This is modeled by the power index attack problem in credit allocation games. 
\label{ex:hirsch}
\end{example} 

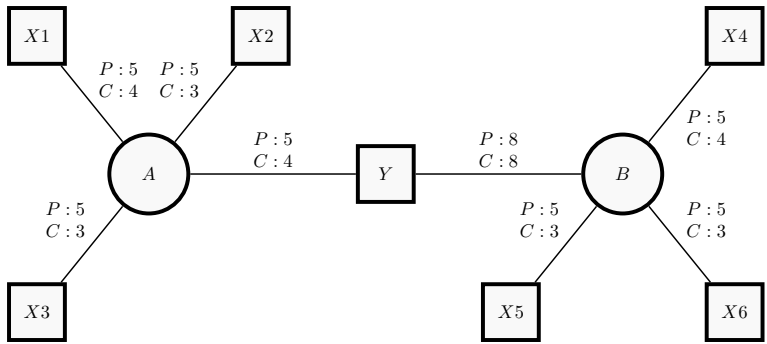
\begin{figure}[h]
\scalebox{0.7}{
\begin{tikzpicture}[node distance=4.5cm, square/.style={regular polygon,regular polygon sides=4},line width=0.75mm, scale=0.3]

	\node[minimum size=1.5cm, draw,circle, fill=gray!5] (A) {$A$};

	\node[minimum size=1.5cm, draw,square,fill=gray!5,right of=A] (Y) {$Y$};

	\node[minimum size=1.5cm, draw,circle,fill=gray!5,right of=Y] (B) {$B$};

	\node[minimum size=1.5cm, draw,square,fill=gray!5,above left=1.5cm and 1cm of A] (X1) {$X1$};

	\node[minimum size=1.5cm, draw,square,fill=gray!5,above right=1.5cm and 1cm of A] (X2) {$X2$};

	\node[minimum size=1.5cm, draw,square,fill=gray!5,below left=1.5cm and 1cm of A] (X3) {$X3$};

	\node[minimum size=1.5cm, draw,square,fill=gray!5,above right=1.5cm and 1cm of B] (X4) {$X4$};

	\node[minimum size=1.5cm, draw,square,fill=gray!5,below left=1.5cm and 1cm of B] (X5) {$X5$};

	\node[minimum size=1.5cm, draw,square,fill=gray!5,below right=1.5cm and 1cm of B] (X6) {$X6$};


	\path[-,draw,thick]

	(A) edge node [above, align=center] {$P:5$\\$C:4$} (Y)

	(Y) edge node [above, align=center] {$P:8$\\$C:8$} (B)

	(A) edge node [above right, align=center] {$P:5$\\$C:4$} (X1)

	(A) edge node [above left, align=center] {$P:5$\\$C:3$} (X2)

	(A) edge node [above left, align=center] {$P:5$\\$C:3$} (X3)

	(B) edge node [below right, align=center] {$P:5$\\$C:4$} (X4)

	(B) edge node [above left, align=center] {$P:5$\\$C:3$} (X5)

	(B) edge node [above right, align=center] {$P:5$\\$C:3$} (X6);

\end{tikzpicture}}

\caption{The scenario (from \cite{papapetrou2011shapley}) in Example~\ref{ex:hirsch}}
\label{hirsch}
\end{figure}

\section{Technical Details}
\label{sthree} 

We will be working in the framework of Algorithmic Cooperative Game Theory, see 
\cite{chalkiadakis2011computational} for a readable introduction. 

We will make use of notation $f\rvert^{b}_{a}$ as a shorthand for $f(b)-f(a)$. 
Given graph $G=(V,E)$ and vertex $v\in V$, we will denote by $N(v)$ the set of neighbors of $V$ and by $\widehat{N(v)}=\{v\}\cup N(v)$. Given $S\subseteq V$, we denote by $\delta(S)$ the set of nodes $y\in V\setminus S$ such that there exists $x\in S$, $(x,y)\in E$. We generalize the setting above to the case when $G$ is a \textit{weighted graph}, i.e. there exists a weight function $w:E\rightarrow \mathbb{R}_{+}$. Given set $S\subseteq V$ and integer $r\geq 1$ we define $B_{w}(S,r)$, \textit{the ball of radius $r$ around $S$}, to be the set $B_{w}(S,r)=\{x\in V: (\exists y\in S)\mbox{ s.t. } d_{w}(x,y)\leq r\}$. We may omit $w$ from this notation when it is simply the graph distance in $G$. 
Also, given "cutoff" distance $d_{cut}$ we define $N_{cut}(x)=B(\{x\},d_{cut})$. 

We will deal with cooperative games with transferable utility, that is pairs $\Gamma=(N,v)$ where $N$ is a set of \textit{players} and $v:\mathcal{P}(S)\rightarrow \mathbb{R}_{+}$ is a \textit{value function} under the partial sets of $S$. If $S\subseteq N$ is a set of players, $v(S)$ is the value that players in coalition $S$ can guarantee for themselves irrespective of the other players' participation. 

Although we could prove similar results for other power indices, e.g. the Banzhaf value, in this paper we restrict ourselves to 
the \textit{Shapley value}. This index measures the portion of the grand coalition value $v(N)$ that a given player $x\in N$ could fairly request for itself. It has the formula \cite{chalkiadakis2011computational}  
$
Sh[v](x) = \frac{1}{n!} \cdot \sum_{\pi\in S_n} [v(S^{x}_{\pi} \cup \{x\}) - v(S^{x}_{\pi})]$, $
\mbox{ where }S^{x}_{\pi} = \{\pi[ i] | \pi[i] \mbox{ precedes x in  } \pi\}$ and $S_n$ is the set of permutation. 

We are concerned with two classes of cooperative games. The first one arose from efforts to define game-theoretic notions of \textit{network centrality} \cite{suri2008determining,aadithya2010game,michalak2013efficient,tarkowski2017game}. 
We define these games as follows:  
\begin{itemize}
\item[-] Game $\Gamma_{NC_1}$ is specified by its value function $v_{NC_1}(S)=|S\cup \delta(S)|$. 
\item[-] Given integer $k\geq 1$, game $\Gamma_{NC_2}$ is specified by its value function $v_{NC_2}(S)=|S\cup \{x\not \in S\mbox{ s.t. } | N(x)\cap S|\geq k \}|$.
\item[-] In game $\Gamma_{NC3}$ graph $G$ is \textit{weighted}. We are also given a positive "cutoff distance" $d_{cut}$. We give the characteristic function $v_{NC_3}$ by $v_{NC_3}(S)=|B(S,d_{cut})|.$
\end{itemize} 

A second class of games, related to the example in \cite{papapetrou2011shapley} is that of \textit{influence-attribution games}, formally defined by Karpov \cite{karpov2014equal}. A \textit{credit-attribution game} is formalized by a set of authors $N = \{1, . . ., n\}$ and a set of publications $P = \{P_1, . . ., P_m\}$. Each paper $P_{j}$ is naturally endowed with a set of \textit{authors} $Auth_j \subseteq N$ and a  \textit{quality score} $w_{j}\in \mathbb{R}_{+}$. In real-life scenarios the quality measure could be 1 (i.e. we simply count papers), a score based on the ranking of the venue the paper was published in, the number of its citations, or even some iterative, PageRank-like variant of the above measures. 

\begin{itemize}
\item[-] \textit{The full credit game} $\Gamma_{FC}$ is specified by its value function $v_{FC}(S)$ which is simply the sum of weights of papers whose authors' list contains \textbf{at least one member from $S$.} 
\item[-] \textit{The full obligation game} $\Gamma_{FO}$ is specified by its value function $v_{FO}(S)$ which is the sum of weights of papers whose authors  \textbf{are all members of $S$.} 
\end{itemize} 

Denote by $Pap_{x}$ the set of papers of $x$, and by $CA(x)$ the set of co-authors of  $x$, i.e. the set of players $l$ for which there exists a $k\in Pap_{x}\cap Pap_{l}$. If $l\in CA(x)$ denote by 
$C(x,l)= \sum\limits_{k \in Pap_x\cap Pap_{l}} w_{k}$ 
\textit{the joint contribution} of  $x,l$. 

\textbf {Reliability extension and attack models}  We will be working within the framework of \textit{reliability extension} of games, first defined in \cite{reliability-games} and further investigated in \cite{bachrach2014cooperative}. The {\em reliability extension} of cooperative game $G=(N,v)$ with parameters $(p_{1},p_{2},\ldots, p_{n})$ is the cooperative game $\Gamma=(N,\overline{v})$ with $\overline{v}(S)=\sum\limits_{T\subseteq S}  v(T)\cdot \Pi_{T,S}$, $\mbox{ where } \Pi_{T,S}=(\prod\limits_{i\in T} p_{i}) \cdot (\prod\limits_{i\in S\setminus T} (1-p_{i})).$ 

A useful  result about these quantities is: 
\begin{claim} Let $S\subseteq W$. We have 
\begin{displaymath} 
\frac{\partial \Pi_{S,W}}{\partial p_{j}}= \left\{\begin{array}{l} 
  \Pi_{S\setminus j,W\setminus j}\mbox{ if }j\in S,\\
  -\Pi_{S,W\setminus j}\mbox{ if }j\in W\setminus S, \mbox{ and } \\
  0, \mbox{ if }j\not \in W
  \end{array}
\right.
\end{displaymath}
\label{aux2}
\end{claim}
\vspace{-0.4cm}

We will consider in the sequel the following two attack models: 
\begin{itemize} 
\item[(1). ]  \textit{fractional attack:} In this type of attack every node $j$ different from the attacked node $x$ has a \textit{baseline reliability} $p^{*}_{j}\in (0,1]$. We are allowed to manipulate the reliability of each such node $j\neq x$ by changing it from $p^{*}_{j}$ to an arbitrary value $p$.  To do so we will incur, however, a cost $u_{j}(p)$. We assume that cost function $u_{j}(\cdot)$ is defined and has an unique zero\footnote{There is no cost for keeping the baseline reliability.} at $p=p^{*}_{j}$, is decreasing and linear on $[0,p^{*}_{j}]$ and increasing and linear on $[p^{*}_{j},1]$ (Figure~\ref{util}).  That is: for every player $j\neq x$ there exist values $L_{j},R_{j}>0$ such that 
\[
u_{j}(p)=\left\{\begin{array}{ll} 
  L_{j}(p^{*}_{j}-p), & \mbox{ if } p<p^{*}_{j}, \\ 
  0, & \mbox{ if } p=p^{*}_{i}, \\ 
  R_{j}(p-p^{*}_{j}), & \mbox{ if } p>p^{*}_{j}.
  \end{array}
\right.
\]
\item[(2). ] \textit{removal attack:} In this type of attack we are only allowed to change the reliability of any node $j$ (different from the targeted node $x$) from $p^{*}_{j}$ to $0$. To do so will incur a cost $c_{j}$.
\end{itemize} 

\textbf{A basis for fractional attacks} The following simple result will be used to analyze fractional attacks in network centrality games:  
\begin{lemma}["IMPROVING SWAPS"]
Let $D$ be an open set in $\mathbb{R}^{n}$, let $x=(x_{1},\ldots, x_{n})\in D$ and $f:D\rightarrow \mathbb{R}$ be an analytic function. Assume $1\leq i,j\leq n$ are indices such that $\frac{\partial f(x_{1},\ldots, x_{n})}{\partial x_{i}} > \frac{\partial f(x_{1},\ldots, x_{n})}{\partial x_{j}}. $  Define $x_{i,j}(\epsilon)=(x_{k}(\epsilon))$, with 
\begin{equation} 
x_{k}(\epsilon)=\left\{\begin{array}{ll} 
  x_{k}+\epsilon, & \mbox{ if } k=j, \\ 
  x_{k}-\epsilon, & \mbox{ if } k=i, \\ 
  x_{k}, & \mbox{otherwise.}
  \end{array}
\right.
\label{change}
\end{equation} 
Then there exists $\epsilon_{0}>0$ such that function 
$g:[0,\epsilon_{0}]\rightarrow \mathbb{R}$, $g(\epsilon)=f(x_{i,j}(\epsilon))$ is monotonically decreasing. 
\label{aux}
\end{lemma} 
In other words, to minimize function $f$ one could decrease the variables with the largest partial derivative, while symmetrically increasing a smaller one. 
\vspace{-3mm}
\begin{proof} 
By the chain rule $g^{\prime}(0)=\sum\limits_{k=1}^n \frac{\partial f(x_{1},\ldots, x_{n})}{\partial x_{k}} \frac{\partial x_{k}(\epsilon)}{\partial \epsilon}|_{\epsilon = 0}$
\begin{align*} 
&  = \frac{\partial f(x_{1},\ldots, x_{n})}{\partial x_{j}} - \frac{\partial f(x_{1},\ldots, x_{n})}{\partial x_{i}}<0. 
\end{align*} 
Since $g^{\prime}$ is continuous, $g^{\prime}$ is strictly negative on some interval $[0,\epsilon_{0}]$. The result follows. 
\end{proof} 


\vspace{-0.7cm}
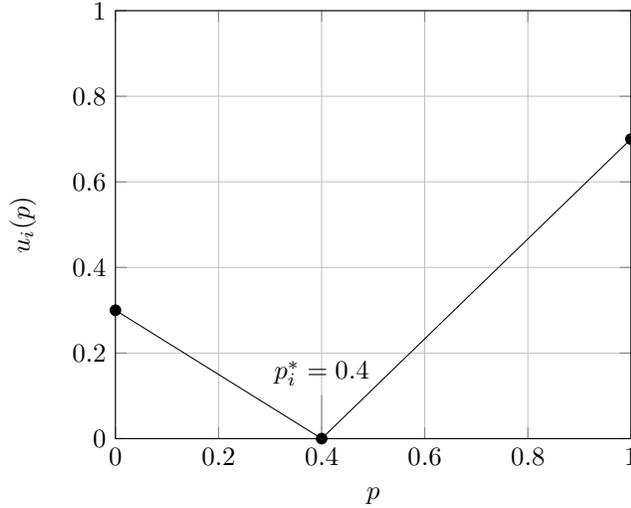
\begin{figure}[ht]
\begin{tikzpicture}[scale = 1]

\begin{axis}[ymin=0, ymax=1, xmin=0, xmax=1, enlargelimits=false, xlabel=$p$, ylabel={$u_{i}(p)$}, grid=both, grid style={line width=.1pt, draw=gray!10}, major grid style={line width=.2pt,draw=gray!50}]

	\addplot[mark=*,black] plot coordinates {(0,0.3) (0.4,0) (1,0.7)};
    \addplot[mark=] coordinates {(0.4,0)} node[pin=90:{$p^{*}_{i}=0.4$}]{} ;
	\end{axis}

\end{tikzpicture}
\vspace{-0.4cm}
\caption{Shape of utility functions in fractional attacks.}
\label{util}
\end{figure} 
\vspace{-0.5cm}

\section{Closed-form formulas} 
\label{sfive} 
The basis for our manipulation of network centralities is the following characterization of the Shapley value of the reliability extension: 

\begin{theorem} The Shapley values of the reliability extensions of network centrality games $\Gamma_{NC_1}, \Gamma_{NC_2},\Gamma_{NC_3}$ have the formulas: 

\begin{displaymath}
Sh[\overline{v_{NC_1}}](x)= p_{x} \mathlarger{\sum}\limits_{\stackrel{y\in \widehat{N(x)}}{S\subseteq \widehat{N(y)}\setminus x}} \frac{1}{|S| + 1} \Pi_{S, \widehat{N(y)}\setminus x} 
\end{displaymath}
\begin{align*}
& Sh[\overline{v_{NC_2}}](x)=p_{x}[\mathlarger{\sum}\limits_{y\in N(x)} \mathlarger{\sum}\limits_{\stackrel{S\subseteq \widehat{N(y)}\setminus x}{|S|\geq k-1}} \frac{(|S|+1-k)}{|S|(|S| + 1)} \Pi_{S, \widehat{N(y)}\setminus x}+\\ 
& + \mathlarger{\sum}\limits_{S\subseteq N(x)}\frac{k}{|S| + 1} \Pi_{S, N(x)} ]\\
& Sh[\overline{v_{NC_3}}](x)= p_{x} \mathlarger{\sum}\limits_{\stackrel{y\in  \widehat{N}(x)}{S\subseteq \widehat{N_{cut}(y)}\setminus x}} \frac{1}{|S| + 1} \Pi_{S, \widehat{N_{cut}(y)}\setminus x} 
\end{align*}
\label{sh-gamma1}
\end{theorem} 
\vspace{-5mm}
As for credit atribution games, the corresponding result is
\begin{theorem} 
The  Shapley values of the reliability extensions of $\Gamma_{FC},\Gamma_{FO}$ with 
probabilities $(p_{1},p_{2},\ldots, p_{n})$ have the formulas 
\begin{equation} \label{sh:rel:1}
Sh[\overline{v_{FC}}](x)=p_{x}\cdot \sum\limits_{k\in Pap_{x}} w_{k}\cdot \big[\sum_{S\subseteq Auth_{k}\setminus \{x\}}\frac{ \Pi_{\emptyset,S} }{(n_{k}-|S|){{n_{k}}\choose {|S|}}} \big]
\end{equation} 
where $Auth_{k}$ is the set of coauthors of paper $k$ and $n_{k}=|Auth_{k}|$, and 
\begin{equation} \label{sh:rel:2}
Sh[\overline{v_{FO}}](x)= \sum\limits_{k\in Pap_{x}} \frac{w_{k}}{n_{k}}\cdot  \Pi_{Auth_{k},Auth_{k}}
\end{equation} 
\label{k=1}
\end{theorem}

\section{Attacking network centralities}

 The next  result follows from Theorem~\ref{sh-gamma1} and Claim~\ref{aux2}: 

\begin{corollary} In the reliability extensions of the centrality games $\Gamma_{NC_1}, \Gamma_{NC_2}, \Gamma_{NC_3}$, the Shapley values of player $1$ are monotonically decreasing functions of distance-two neighbors' reliabilities (and do not depend on other players). 
\label{centrality} 
\end{corollary}  
\begin{proof} 
Deferred to the full version. 
\end{proof} 

The previous corollary shows that for network centrality games no removal attack is 
beneficial: 

\begin{theorem} 
No removal attack on the centrality of a player in games $\Gamma_{NC_1}, \Gamma_{NC_2}, \Gamma_{NC_3}$ can decrease its Shapley value. 
\end{theorem} 

\textbf{Fractional attacks on specific networks}  Given that removal attacks are not beneficial, we now turn to fractional attacks. The objective of this section is to show that 
the analysis of optimal fractional attacks is often feasible. Since the graphs in this section are fairly symmetric, we will assume 
(for these examples) that the slopes of all utility curves are identical. That is, there exist positive constants $L,R$ such that if $i\neq j$ are different agents then  $L_{i}=L_{j}=L$ and $R_{i}=R_{j}=R$ (though, of course, baseline probabilities $p^{*}_{i}$ and $p^{*}_{j}$ may differ). 
The graphs we are going to be concerned with are the complete graph $K_{n}$, the star graph $S_{n}$ (where node 1 is either the center or an outer node) and the $n$-cycle $C_{n}$ (Figure~\ref{fig:completeGraph}).

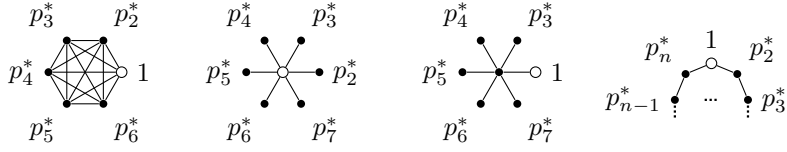
\begin{figure}[ht]
\hspace{-3mm} 
\begin{minipage}{.22\textwidth}
\begin{center} 
\begin{tikzpicture}[scale=0.090, ->,>=stealth',level/.style={sibling distance = 3cm/#1,
  level distance = 1cm}] 
  
  \def \n {6}
  
  \pgfmathparse{1 * (360 / \n) - (360 / \n)};
  \node[circle,draw,inner sep=0pt,minimum size=4pt,label=\pgfmathresult:1] (N-1) at (\pgfmathresult:5.4cm){};
    
  \foreach \x in {2,...,\n}{
    \pgfmathparse{\x * (360 / \n) - (360 / \n)};
    \node[circle,fill=black,inner sep=0pt,minimum size=3pt,label=\pgfmathresult:{$p^{*}_{\x}$}] (N-\x) at (\pgfmathresult:5.4cm){};
  }; 
  
  \foreach \x [count=\xi from 1] in {2,...,\n}{%
    \foreach \y in {\x,...,\n}{%
        \path (N-\xi) edge [-] (N-\y);
  	}
  }
  
\end{tikzpicture}
\end{center} 
\end{minipage}
\hspace{-2mm}
\begin{minipage}{.22\textwidth}
\begin{center} 
\begin{tikzpicture}[scale=0.090, ->,>=stealth',level/.style={sibling distance = 3cm/#1,
  level distance = 1cm}] 
  
  \def \n {6}
  
  \foreach \x [count=\xi from 2] in {1,...,\n}{
    \pgfmathparse{\x * (360 / \n) - (360 / \n)};
    \node[circle,fill=black,inner sep=0pt,minimum size=3pt,label=\pgfmathresult:{$p^{*}_{\xi}$}] (N-\x) at (\pgfmathresult:5.4cm){};
  };
  
  \node[circle,draw,inner sep=0pt,minimum size=4pt,] (N-0) at (0:0cm){}; 
  
  \foreach \x in {1,...,\n}{%
        \path (N-0) edge [-] (N-\x);
  }
\end{tikzpicture}
\end{center} 
\end{minipage}
\hspace{-1mm}
\begin{minipage}{.22\textwidth}
\begin{center} 
\begin{tikzpicture}[scale=0.090, ->,>=stealth',level/.style={sibling distance = 3cm/#1,
  level distance = 1cm}] 
  
  \def \n {6}
\pgfmathparse{1 * (360 / \n) - (360 / \n)};
\node[circle,draw,inner sep=0pt,minimum size=4pt,label=\pgfmathresult:1] (N-1) at (\pgfmathresult:5.4cm){};
    
  \foreach \x [count=\xi from 3] in {2,...,\n}{
    \pgfmathparse{\x * (360 / \n) - (360 / \n)};
    \node[circle,fill=black,inner sep=0pt,minimum size=3pt,label=\pgfmathresult:{$p^{*}_{\xi}$}] (N-\x) at (\pgfmathresult:5.4cm){};
  };

  \node[circle,fill=black,inner sep=0pt,minimum size=3pt,] (N-0) at (0:0cm){};
  
  \foreach \x in {1,...,\n}{%
        \path (N-0) edge [-] (N-\x);
  }
  
\end{tikzpicture}
\end{center} 
\end{minipage}
\hspace{-2mm}
\begin{minipage}{.22\textwidth}
\begin{center} 
\begin{tikzpicture}[scale=0.090, ->,>=stealth',level/.style={sibling distance = 3cm/#1,
  level distance = 1cm}] 
  
  \def \n {8}
  
  \foreach \x/\l in {1/$p^{*}_3$, 2/$p^{*}_2$, 3/1, 4/$p^{*}_{n}$, 5/$p^{*}_{n-1}$}{        
	\pgfmathparse{\x * (360 / \n) - (360 / \n)};
    		    
    \ifthenelse{\equal{\l}{1}}
    {
	    	\node[circle,draw,inner sep=0pt,minimum size=4pt,label=\pgfmathresult:\l] (N-\x) at (\pgfmathresult:5.4cm){}; 
    }
    {
	    	\node[circle,fill=black,inner sep=0pt,minimum size=3pt,label=\pgfmathresult:\l] (N-\x) at (\pgfmathresult:5.4cm){}; 
  	}
  };
  
  \foreach \x [count=\xi from 1] in {2,...,5}{%
     \path (N-\x) edge [-] (N-\xi);
  }
  
  \path (N-1.south) + (0, -1.95) coordinate(x1) edge[line width=0.85pt,-,densely dotted] (N-1);
  \path (N-5.south) + (0, -1.95) coordinate(x2) edge[line width=0.85pt,-,densely dotted] (N-5);
   \path (-1, 0) coordinate(x3) edge[line width=0.85pt,-,densely dotted] (1,0);
\end{tikzpicture}
\end{center}
\end{minipage}

  \caption{Target topologies for fractional attacks.}
  \label{fig:completeGraph}
\end{figure}
\vspace{-0.2cm}

\vspace{2mm}
Note that, when $G=K_{n}$ or $G=S_{n}$, pairwise Shapley value attacks  are trivially impossible: indeed, these graphs have diameter at most two. Since all distance-two neighbors influence the Shapley value of a given player, all nodes are exempt from attacks. 

On the other hand, for these topologies  it turns out that the best attack on Shapley value of player $x$ is to increase the reliabilities of its neighbors in the descending order of their baseline reliabilities: 

\begin{theorem} Let $G$ be either the complete graph $K_{n}$ with $n$ vertices.
or the star graph with $n$ vertices $S_{n}$ centered at node $x=1$. To optimally attack the centrality of $x$ in the reliability extension of $\Gamma_{NC_1}$  use the following algorithm: 

\begin{framed} 
\noindent- Consider nodes $2,\ldots, n$ in the decreasing order of their baseline reliabilities, breaking ties arbitrarily. $p^{*}_{sorted(2)}\geq p^{*}_{sorted(3)}\geq \ldots \geq p^{*}_{sorted(n)}.$

\noindent- While the budget allows it, increase to one (if not already equal to 1) the probabilities $p_{sorted(i)}$, starting with $i=2$ and successively increasing $i$. 

\noindent- If the budget no longer allows increasing $p_{sorted(i)}$ to one, increase it as much as possible. 

\noindent- Leave all other probabilities to their baseline values. 
\end{framed}
\noindent If, on the other hand, $G=S_{n}$ centered, say, at node 2, to optimally attack the centrality of node $x=1$, the algorithm changes as follows: 

\begin{framed} 
\noindent - Consider nodes $2,\ldots, n$ in the following order: node 2, followed by nodes $3,\ldots, n$ sorted in decreasing order of their baseline reliabilities $p^{*}_{sorted(3)}\geq \ldots \geq p^{*}_{sorted(n)}$, breaking ties arbitrarily. Denote the new order by 
$Q$. 

\noindent - Follow the previous greedy protocol, increasing baseline probabilities up to one (if allowed by the budget) according to  the new ordering $Q$. 
\end{framed}

Similar statements hold for game $\Gamma_{NC_2}$, and for  $\Gamma_{NC_3}$ for large enough values of parameter $d_{cut}$.  
\label{kn}
\end{theorem} 

\noindent 

In the previous examples the optimal attack involved a determined node targeting order, which privileged direct neighbors and could depend on baseline reliabilities but  was independent of the value of the budget. None of this holds in general:  as the next result shows, on graph $C_{n}$ the optimum can be computed by taking the best of \textit{four} node targeting orders. The optimum may lack the two previously discussed properties of optimal orders: 
\begin{itemize} 
\item[-] in optimal attacks one should sometimes target a distance-two neighbor (3 or n-1) \textit{before} targeting both of $x=1$'s  neighbors (2 and $n$, see Figure~\ref{fig:completeGraph}). 
\item[-] the order (among the four) that characterizes the optimum may depend on the budget value $B$ as well. Formally: 
\end{itemize} 

\begin{theorem} Let $P,Q,R,S$ be the vectors $[2,n,n-1,3]$, $[2,n-1,n,3]$, $[n,3,2,n-1]$, $[n,2,3,n-1]$, respectively. Let $Sol_{P},Sol_{Q},Sol_{R}$, $Sol_{S}$ be the configurations obtained by increasing in turn (as much as possible, subject to the budget $B$) the reliabilities of nodes $2,3,n-1,n$ in the order(s) specified by $P,Q,R,S$, respectively. Then 
\begin{itemize} 
\item[a.] The best of $Sol_{P},Sol_{Q},Sol_{R},Sol_{S}$ is an  
optimal attack on the centrality of  $x=1$ in game $\Gamma_{NC_{1}}$ on the cycle graph $C_{n}$. 
\item[b.] There exist values of $p_{2}^{*},p_{3}^{*},p_{n-1}^{*},p_{n}^{*}$ s.t.  $Sol_{P}$ is optimal for all values of $B$ (by symmetry a similar statement holds  for $Sol_S$). 
\item[c.] There exist values of $p_{2}^{*},p_{3}^{*},p_{n-1}^{*},p_{n}^{*}$ and an nonempty open interval $I$ for the budget $B$ such that $Sol_{Q}$ is an optimum for all $B\in I$ (by symmetry a similar statement holds for $Sol_R$).
\end{itemize}  
 \label{cn}
 \end{theorem} 
\section{Attacks in credit attribution} 
\label{ssix} 

In this section we study removal attacks in credit attribution games. Interestingly, while the Shapley values  have identical formulas in $\Gamma_{FC},\Gamma_{FO}$ \cite{karpov2014equal}, the  two games are \textbf{not} similar with respect to attacks. Indeed, similarly to the case of network centrality, we have: 

\begin{theorem} 
No removal attack can decrease the Shapley value of a given player in a full credit attribution game. 
\label{full-credit} 
\end{theorem}
\begin{proof} 
At first, this seems counterintuitive, as it would seem to contradict Example~\ref{ex:hirsch}. The answer is that \textit{this example does not correspond to the full credit game, but to the full obligation one}: in game $\Gamma_{FC}$ a player does \textbf{not} lose credit for a paper due to removal of a coauthor; in fact its Shapley value will increase, since the credit for the paper divides among fewer coauthors. It is in $\Gamma_{FO}$ where players may lose credit as a result of coauthor removal.
\end{proof}
This difference between $\Gamma_{FC}$ and $\Gamma_{FO}$  is evident with respect to  attacks: As the next result shows, in full-obligation games, finding optimal  removal attacks can simulate a well-known hard combinatorial problem: 

\begin{theorem} The {\em budgeted maximum coverage problem} (which is NP-complete) reduces to minimizing the Shapley value of a given player in the full-obligation game (under removal attacks).  
\label{full-obligation}  
\end{theorem}
\begin{proof} 
Deferred to the full version. 
\end{proof} 

\textbf{Fractional attacks} The following is a simple consequence of the formulas in Theorem~\ref{k=1} and Claim~\ref{aux2} shows that \textit{optimal attacks are different in games $\Gamma_{FC}$ and $\Gamma_{FO}$ irrespective of the topology of the coauthorship hypergraph}: in the first case we need to increase the reliability of $x$'s coauthors, in the other case we aim to decrease it: 

\begin{theorem} 
In the reliability extensions of the credit allocation games $\Gamma_{FC}, \Gamma_{FO}$ the Shapley value of player $x$ is a decreasing (respectively increasing) function of coauthors' reliabilities (and does not depend on other players). 
\label{credit}
\end{theorem}  

Optimal attacks can be explicitly described in the particular scenario when, just as in Example~\ref{ex:hirsch}, each paper has exactly two authors (a situation that corresponds, under the full obligation model, to  induced subgraph games). It turns out that \textit{the relevant quantity is the ratio between the score of  coauthors' joint contribution with the attacked node and its marginal cost:} 

\begin{theorem} 
To optimally decrease the Shapley value of node $x$ in game $\Gamma_{FC}$ in the two-author special case: 
\begin{framed} 
\noindent (a). Sort the coauthors $l$ of $x$ in the decreasing order of 
the fractions $\frac{C(x,l)}{R(l)}$, breaking ties arbitrarily. 

\noindent (b). While the budget allows it,  for $i=1,\ldots |CA(x)|$, \textbf{increase} to 1 the probability of the $i$'th most valuable coauthor.   

\noindent (c). If the budget does not allow increasing  the probability of the $i$'th coauthor up to 1, increase it as much as possible. 

\noindent (d). Leave all other probabilities to their baseline values. 
\end{framed} 
\label{kn2}
\end{theorem} 

\vspace{-0.3cm}
\begin{corollary} In the setting of Theorem~\ref{kn2}, 
to optimally solve the pairwise Shapley value attack problem for  $x,y$, run the algorithm in the Theorem only on those  $z$ that are coauthors of $x$ but not of $y$. 
\end{corollary} 

As for game $\Gamma_{FO}$, the optimal attack is symmetric. Since we are decreasing probabilities, we will  be using fractions $\frac{C(x,l)}{L(l)}$ instead: 

\begin{theorem} 
To optimally decrease the Shapley value of node $x$ in the full obligation game $\Gamma_{FO}$ in the two-author special case: 
\begin{framed} 
\noindent (a). Sort the coauthors of $x$ in the decreasing order of the fractions $\frac{C(x,l)}{L(l)}$, breaking ties arbitrarily. 

\noindent (b). While the budget allows it, for $i=1,\ldots |CA(x)|$,  \textbf{decrease}  to 0 the probability of the $i$'th most valuable coauthor. 

\noindent (c). If the budget does not allow decreasing  the probability of the $i$'th coauthor up to 0, decrease it as much as possible. 

\noindent (d). Leave all other probabilities to their baseline values. 
\end{framed} 
 \label{cn2}
 \end{theorem} 
 
\vspace{-0.5cm} 
 \begin{corollary} In the setting of Theorem~\ref{cn2},
to solve the pairwise Shapley value attack problem for players $x,y$, run the algorithm in the Theorem only on  those  $z$ that are coauthors of $x$ but not of $y$. 
\end{corollary} 

\section{Proof Highlights} 

In this section we present some of the proofs of our results. Some other proofs are included in the Appendix, others are deferred to the full version of the paper, to be posted on arxiv.org: 

\subsection{Proof of Theorem~\ref{sh-gamma1}}

We prove the formula for the first game only. Similarly to \cite{michalak2013efficient}, the proofs for the other two games are completely analogous, and deferred to the full version. Define, for $y\in V$, $W\subseteq V$ 
\[
f_{y}(W)=\left\{\begin{array}{ll} 
  1, & \mbox{ if } y\not \in W\cup \delta(W), \\ 
  0, & \mbox{ otherwise.}  \\ 
  \end{array}
\right.
\]
A simple case analysis proves that, for every $W\subseteq V$, $
 v_{NC_1}(W\cup \{x\})-v_{NC_1}(W)= \sum_{y\in \widehat{N(x)}} f_{y}(W).$
We therefore have
\begin{align*}
& Sh[\overline{v_{NC_1}}](x) = E_{\pi \in S_{n}}[\overline{v_{NC_1}}(S_{\pi}^{x} \cup \{x\}) - \overline{v_{NC_1}}(S_{\pi}^{x})] = E_{\pi\in S_{n}}[p_{x} \cdot \\
&\sum\limits_{W\subseteq S_{\pi}^{x}} [v_{NC_1}(W \cup \{x\}) - v_{NC_1}(W)] \cdot \Pi_{W, S_{\pi}^{x}}]= p_{x} E_{\pi \in S_{n}}  \sum\limits_{W\subseteq S_{\pi}^{x}}\\
& \sum\limits_{y \in \widehat{N(x)}} f_{y}(W) \cdot \Pi_{W, S_{\pi}^{x}} = p_{x} E_{\pi\in S_{n}} \sum\limits_{y \in \widehat{N(x)}} \sum\limits_{W\subseteq S_{\pi}^{x}}  f_{y}(W) \cdot \Pi_{W, S_{\pi}^{x}}
\end{align*}
We now introduce two notations that will help us reinterpret the previous sum: given $W\subseteq V$, denote  by $Alive(W)$ the set of nodes in $W$ that are \textit{alive} under the reliability extension model. Also, given permutation $\pi\in S_{n}$ and $W\subseteq V$, denote by 
$First_{\pi}(W)$ the element of $W$ that appears first in enumeration $\pi$. With these notations
\begin{align*} 
&Sh[\overline{v_{NC_1}}](x) = p_{x} \sum\limits_{y \in \widehat{N(x)}} Pr_{\pi\in S_{n}}[y \notin Alive(S_{\pi}^{x}) \cup \delta(Alive(S_{\pi}^{x}))]  \\
&= p_{x} \sum\limits_{y \in \widehat{N(x)}} Pr_{\pi\in S_{n}}[x = First_{\pi}(\widehat{N(y)} \cap Alive(V))| x \in Alive(V) ] 
\end{align*}
If $S=(\widehat{N(y)}\setminus\{x\}) \cap Alive(V)$ then the conditional probability that $x$ is $First_{\pi}(S\cup \{x\})$, given that $x$ is alive, is $\frac{1}{|S|+1}$. We thus get the desired formula.

\subsection{Proof of Theorem~\ref{k=1}}

Denote, for a set of authors $C$, by $Pap_{C}=\cup_{l\in C} Pap_{l}$ the set of papers with at least one author in $C$. 
We decompose function $v_{FC}$ as  $v_{FC}(\cdot)=\sum_{k} w_{k}v_{k}(\cdot)$ where 
\begin{equation} 
v_{k}(C)=\left\{\begin{array}{l} 
  1,\mbox{ if } k\in Pap_{C} \\ 
  0, \mbox{otherwise.}
  \end{array}
\right.
\end{equation} 
\begin{align*}
\mbox{ Thus }v_{FC}(C)  & = \sum\limits_{R\subseteq C} v_{FC}(R)\Pi_{R,C} =  \sum\limits_{R\subseteq C}  \Pi_{R,C} \sum_{k} w_{k}v_{k}(R) = \\ 
 & = \sum_{k} \sum\limits_{R\subseteq C}\Pi_{R,C}  w_{k}v_{k}(R)= \sum_{k} w_{k}\overline{v_{k}}(C)
\end{align*} 

which means that we can decompose $\overline{v_{FC}}=\sum_{k} w_{k}\overline{v_{k}}$, and  the Shapley value of $\overline{v_{FC}}$ decomposes as well $
Sh(\overline{v_{FC}})=\sum_{k} w_{k}\cdot Sh(\overline{v_{k}}),  $
and similarly for $v_{FO}$. On the other hand
\[
Sh[\overline{v_{k}}](x)=\frac{1}{n!}\sum\limits_{\pi \in S_{n}} [\overline{v_{k}}(S^{x}_{\pi}\cup \{x\})- \overline{v_{k}}(S^{x}_{\pi})]
\]
Given set $A$ of authors,
\begin{align*}
& \overline{v_{k}}(A\cup \{x\})-\overline{v_{k}}(A)  = \sum\limits_{R\subseteq A\cup\{x\}} v_{k}(R) \Pi_{R,A\cup \{x\}}-  \sum\limits_{R\subseteq A} v_{k}(R)\Pi_{R,A}  \\
& = (1-p_{x})\sum\limits_{R\subseteq A\setminus x} v_{k}(R)\Pi_{R,A\setminus \{x\}}+ 
p_{x}\sum\limits_{R\subseteq A\setminus x} \Pi_{R,A} v_{k}(R\cup \{x\}) \\
&- \sum\limits_{R\subseteq A} \Pi_{R,A} v_{k}(R) = p_{x}\cdot \sum\limits_{R\subseteq A} \Pi_{R,A} [v_{k}(R\cup\{x\})-v_{k}(R)]
\end{align*}
Now $v_{k}(R\cup\{x\})-v_{k}(R)$ is 1 if $k\in Pap_{x}\setminus Pap_{R}$, 0 otherwise. For 
$k\not \in Pap_{x}$, $\overline{v_{k}}(A\cup \{x\})-\overline{v_{k}}(A)=0$. Otherwise 
$\overline{v_{k}}(A\cup \{x\})-\overline{v_{k}}(A)  = p_{x}\cdot \sum\limits_{\stackrel{R\subseteq A}{k\not \in Pap_{R}}} \Pi_{R,A}.$

We can interpret this quantity as the probability that the live subset of $A$ does not cover $k$, but $x$ is live and does.  Applying this to the Shapley value we infer that $Sh[\overline{v_{k}}](x)$ is the probability that in a random permutation $\pi$ the  live subset of $S^{x}_{\pi}$ does not cover $k$, but $x$ is live and does. 

\textbf{Full credit model:}  There are $n_{k}!$ permutations $\Xi$ of indices in $Auth_{k}$, each of them equally likely when $\pi$ is a random permutation in $S_{n}$. Given subset $S\subset Auth_{k}\setminus \{x\}$, the probability that $\Xi$ starts with $S$ followed by $x$ is 
$\frac{|S|!(n_{k}-|S|-1)!}{n_{k}!}$. To make $x$ pivotal for paper $k$, none of the agents in $S$ must be live. This happens with probability $\Pi_{\emptyset,S}$.  Given the above argument,  we have 
\begin{align*}
Sh[\overline{v_{k}}](x)& = p_{x}\cdot  \sum_{S\subseteq Auth_{k}\setminus \{x\}} \frac{(|S|)!(n_{k}-|S|-1)!}{n_{k}!}\cdot [\prod_{l\in S} (1-p_{l})] \\
& =p_{x}\cdot \sum_{S\subseteq Auth_{k}\setminus \{x\}}\frac{ \Pi_{\emptyset,S}}{(n_{k}-|S|){{n_{k}}\choose {|S|}}}, \mbox{ hence }
\end{align*}  
\begin{equation} 
Sh[\overline{v_{FC}}](x)=p_{x}\cdot \sum\limits_{k\in Pap_{x}} w_{k}\cdot [\sum_{S\subseteq Auth_{k}\setminus \{x\}}\frac{ \Pi_{\emptyset,S} }{(n_{k}-|S|){{n_{k}}\choose {|S|}}}]
\end{equation} 

which is what we had to prove.

\textbf{Full obligation model:} For $x$ to be pivotal for paper $k$, $x$ and all its coauthors in $Auth_{k}$ must all be live, and all elements of $Auth_{k}\setminus x$ must appear before $x$ in ordering $\pi$. This happens with probability 
$\frac{1}{n_{k}}\cdot \Pi_{Auth_{k},Auth_{k}}.$

\subsection{Proof of Theorem~\ref{kn}}

First of all, the following claim holds for all graphs $G$: 
\begin{claim}
The minimum of function $z\rightarrow Sh[\overline{v_{NC_1}}](1)|_{z}$ exists and is reached on some profile $(p_{i})$ with $p^{*}_{i}\leq p_{i}\leq 1$. 
\end{claim} 
\begin{proof} 
Function $z\rightarrow Sh[\overline{v_{NC_1}}](1)|_{z}$ is continuous and the set $[0,1]^{n}$ is compact, so the minimum is reached. Assuming some $p_{j}<p^{*}_{j}$,  we could increase $p_{j}$ up to $p^{*}_{j}$,  reducing total cost. This does not increase (and perhaps further decreases) the Shapley value. 
\end{proof} 

Next, we (jointly) prove cases $G=K_{n}$ and $G=S_{n}$ with $x=1$ being a center, since the proofs are practically identical. The remaining case ($K=S_{n}$, $x=1$ not a center) is deferred to the Appendix. We start with the following 

\begin{lemma} 
For $G=K_n$ or $G=S_{n}$, $j\neq l\in V(G)\setminus 1$ and any probability profile $p=(p_{1},\ldots, p_{n})\in (0,1]^{n}$, 
\[
sign\big( \frac{\partial Sh[\overline{v_{NC_1}}](1)}{\partial p_{j}}|_{p} -  \frac{\partial Sh[\overline{v_{NC_1}}](1)}{\partial p_{l}}|_{p} \big)=sign(p_{j}-p_{l})
\]
\label{sign} 
\end{lemma} 

\vspace{-0.8cm}
\begin{proof}
Deferred to the full version. 
\end{proof} 

\vspace{-0.3cm}
We first prove that in the optimal solution on these graphs no two variables could assume equal values, unless both equal to the endpoints of their restricting intervals: 
\begin{lemma}
In the setting of Theorem~\ref{kn}, suppose $p=(p_{1},\ldots, p_{n})$ is such there is are indices $2\leq i\neq j\leq n$ with $0<p_{i}=p_{j}<1$. Then there exists $\epsilon_{0}>0$ such that for every $\epsilon \in [-\epsilon_{0},\epsilon_{0}]$, $\epsilon \neq 0$, $
Sh[\overline{v_{NC_1}}](1)|_{p_{i,j}(\epsilon)}< Sh[\overline{v_{NC_1}}](1)|_{p}$, 
(where $p_{i,j}(\epsilon)$ is defined as in equation~(\ref{change})). 
\label{equal}
\end{lemma} 
\begin{proof} 
Deferred to the full version. 
\end{proof} 

Now we prove: 

\begin{claim} 
In the optimal solution there is at most one index $i_{1}$ with  $p_{i_{1}}\in (p_{i_1}^{*},1)$. In other words,  in the optimal solution some probabilities are increased up to 1, some ae left unchanged to their baseline values, and at most one variable is increased to a value less than 1.  
\label{atmostone}
\end{claim} 
\begin{proof} 
Suppose there were two different indices $i_{1}\neq i_{2}$. We must have $p_{i_{1}}=p_{i_{2}}$, or, by Lemma~\ref{aux}, one could decrease the Shapley value by increasing the larger one and symmetrically decreasing the smaller one. But this is impossible, due to Lemma~\ref{equal}. 
\end{proof}

Note that the greedy solution $\Gamma$ has the structure from Claim~\ref{atmostone} and that any permutation of OPT on variables $p_{2},\ldots, p_{n}$ has the same Shapley value as OPT (since $K_n,S_n$ have this symmetry). 

We compare the vectors $\Gamma,OPT$, both sorted in decreasing order. Our goal is to show that these sorted 
versions are equal.   Without loss of generality, we may assume that $OPT$ creates the same ordering on  variables as the $p_{i}^{*}$'s (and $\Gamma$), when considered in decreasing sorted order (we break ties, if any, in the same way). Indeed, if there were indices $i,j$ such that $p_{sorted(i)}^{*}\geq p_{sorted(j)}^{*}$ but $p_{sorted(i)}<p_{sorted(j)}$ then, since  $p_{sorted(j)}>p_{sorted(i)}\geq p_{sorted(i)}^{*}\geq p_{sorted(j)}^{*}$, we could simply swap values $p_{sorted(i)}$ and $p_{sorted(j)}$ and obtain another legal, optimal solution. 

If $\Gamma$ were different from $OPT$, since Greedy increases the largest variables first, there must be variables $x,y$ such that $\Gamma_{x}\geq \Gamma_{y}$, $\Gamma_{x}>p_{x}$ and $\Gamma_{y}<p_{y}$. Since $\Gamma$ and $OPT$ have the same ordering of variables, we also must have in fact
$p_{x}\geq p_{y}$, i.e. $1\geq \Gamma_{x}>p_{x}\geq p_{y}>\Gamma_{y}\geq p^{*}_{y}$. But then, using either Lemma~\ref{aux} (if $p_{x}\neq p_{y}$) or Lemma~\ref{equal} (otherwise) we could further improve $OPT$ by increasing $p_{x}$ and symmetrically decreasing $p_{y}$, a contradiction. 

\subsection{Proof of Theorem~\ref{cn}}
A simple computation shows that for $G=C_{n}$
\[
Sh[\overline{v_{NC_1}}](1)=p_1(\frac{p_{2}p_{n}+p_{2}p_{3}+p_{n-1}p_{n}}{3}-\frac{p_{3}+p_{n-1}}{2}- p_{2}-p_{n}+3). 
\]

As $p_1$ does not influence any attack on itself, w.l.o.g. we will assume $p_1 = 1$.
We need to  minimize the above quantity, subject to 
\begin{align*}
& p_{2}+p_{3}+p_{n-1}+p_{n}=B+p^{*}_{2}+p^{*}_{3}+p^{*}_{n-1}+p^{*}_{n}, p^{*}_{i}\leq p_{i}\leq 1. 
\end{align*}
We now prove a result somewhat similar to Claim~\ref{atmostone}. However, now we will only interdict certain patterns. 
\begin{claim} 
In an optimal solution it is not possible that $p^{*}_{k}< p_{k}< 1$, $p^{*}_{l}< p_{l}< 1$ when: 
\begin{itemize} 
\item[a.] $k=2$, $l=n-1$ (and, symmetrically, $k=3$, $l=n$).  In fact, in this case we have the stronger implication $p_{n-1}>p_{n-1}^{*}\Rightarrow p_{2}=1$. Symetrically, $p_{3}>p_{3}^{*}\Rightarrow p_{n}=1$. 
\item[b.] $k=2$, $l=n$. 
\item[c.] $k=2$, $l=3$ (and, symmetrically, $k=n$, $l=n-1$.) In the case when $\frac{p_{3}+p_{n}}{3}\leq \frac{p_{2}}{3}-\frac{1}{2}$ we have the stronger implication $p_{3}> p_{3}^{*}\Rightarrow p_{2}=1$. Symetrically, in the case when $\frac{p_{2}+p_{n-1}}{3}\leq \frac{p_{n}}{3}-\frac{1}{2}$,  $p_{n-1}>p_{n-1}^{*}\Rightarrow p_{n}=1$. 
\end{itemize} 
\label{foo}
\end{claim} 

\begin{proof} 
Suppose there were two such indices $k,l$. We must also have $\frac{\partial Sh[\overline{v_{NC_1}}](1)}{\partial x_{k}} = \frac{\partial Sh[\overline{v_{NC_1}}](1)}{\partial x_{l}}$, otherwise we could decrease the Shapley value using Lemma~\ref{aux}. We reason in all cases by contradiction: 

a. We prove directly the stronger result. Suppose $p_{2}<1$. We have $\frac{\partial Sh[\overline{v_{NC_1}}](1)}{\partial x_{2}}= \frac{p_{3}+p_{n-1}}{3}-1\leq \frac{p_{3}}{3}-\frac{2}{3}< \frac{p_{3}}{3}-\frac{1}{2}=\frac{\partial Sh[\overline{v_{NC_1}}](1)}{\partial x_{n-1}}$. 
So we can apply Lemma~\ref{aux} to $p_{2}$ and $p_{n-1}$, further decreasing the Shapley value as we increase $p_{2}$ and decrease $p_{n-1}$. 

b. Equality of partial derivatives can be rewritten as $p_{2}+p_{n}=p_{3}+p_{n-1}$. An easy computation (which uses this equality) shows that $Sh[\overline{v_{NC_1}}](1)|^{p_{n,2}(\epsilon)}_{p}=-\frac{\epsilon^2}{3}$. But then it means that one could further decrease 
the Shapley value of player 1, hence we are not at an optimum, a contradiction.
 
c.  
As in the proof of a. $\frac{p_{3}+p_{n}-p_{2}}{3}-\frac{1}{2}= 
\frac{\partial Sh[\overline{v_{NC_1}}](1)}{\partial x_{2}}- \frac{\partial Sh[\overline{v_{NC_1}}](1)}{\partial x_{3}}$ $=0$, otherwise we could use Lemma~\ref{aux} with $p_{2},p_{3}$ to decrease the Shapley value.  An easy computation (which uses this equality) shows that
$Sh[\overline{v_{NC_1}}](1)|^{p_{3,2}(\epsilon)}_{p}=\frac{\epsilon(p_{n}-p_{2}+p_{3})}{3}-\frac{\epsilon}{2}-
\frac{\epsilon^2}{3}= -\frac{\epsilon^2}{3}<0$.  But then one could further decrease 
the Shapley value of 1,  a contradiction. 

\end{proof} 

\vspace{-0.5cm}
We use Claim \ref{foo} to prove Theorem~\ref{cn}: 

\textbf{a.} The conclusion of this claim is that the only case when there could exist two values $p_{k},p_{l}$ strictly between their baseline values and 1 is $k=3,l=n-1$ (or vice-versa), a case when we must further have $p_{2}=p_{n}=1$. Thus the optimal solution is the best of the configurations obtained by greedily increasing probabilities (up to 1, if the budget will allow it) in one of the orders 
$[2,n,3,n-1],[2,n,n-1,3],[2,n-1,n,3],[n,3,2,n-1], [n,2,3,n-1],[n,2,n-1,3]$. An easy computation shows that the first two orders are equally good for all possible budget values $B$, and so are the last two. So, in the end we only have to compare the four orders $P,Q,R,S$ to find an optimum, proving the first part of the theorem. 

\textbf{b,c:} Symmetry between 2,3 and n,n-1 reduces the proof of these two points to analyzing the ``winners'' among $Sol_{P},Sol_{Q},Sol_{R}$, $Sol_{S}$, and proving that, under suitable conditions, it belongs either to $\{Sol_{P},Sol_{S}\}$ (point b.) or to 
$\{Sol_{Q},Sol_{R}\}$ (point c.). 

If we start by increasing $p_{2}$ by $\epsilon$, the Shapley value decreases by $\epsilon(1-\frac{p_{3}^{*}+p_{n-1}^{*}}{3})$. 
We will call the number $1-\frac{p_{3}^{*}+p_{n-1}^{*}}{3}$ 
the \textit{speed of the decrease}. It is maintained while $p_{2}$ increases from $p_{2}^{*}$ to 1, i.e. over a \textit{segment} (interval) of \textit{size} $1-p_{2}^{*}$. There are four segments, corresponding to the four variables being increased. The table in Figure~\ref{dec} summarizes the effect of variable increases on the decrease of the Shapley value of node 1. Using this table it is easy to compare the four permutations with respect to this decrease: 

\begin{figure} 
\scalebox{0.7}{
\begin{tabular}{|c|c|c|c|c|c|c|c|c|}
\hline
Perm, & Sp$_{1}$ & sz$_{1}$ & Sp$_{2}$ & sz$_{2}$ & Sp$_{3}$ & sz$_{3}$ & Sp$_{4}$ & sz$_{4}$\\
\hline
$P$ & $1-\frac{p_{3}^{*}+p_{n}^{*}}{3}$  & $1-p_{2}^{*}$ & $\frac{2-p_{n-1}^{*}}{3}$& $1-p_{n}^{*}$ & 1/6 & $1-p_{n-1}^{*}$ &1/6 & $1-p_{3}^{*}$  \\
\hline
$Q$ & $1-\frac{p_{3}^{*}+p_{n}^{*}}{3}$ & $1-p_{2}^{*}$ & $\frac{1}{2}-\frac{p_{n}^{*}}{3}$ & $1-p_{n-1}^{*}$ & 1/3 & $1-p_{n}^{*}$ &1/6 & $1-p_{3}^{*}$  \\
\hline
$R$ & $1-\frac{p_{2}^{*}+p_{n-1}^{*}}{3}$ & $1-p_{n}^{*}$ & $\frac{1}{2}-\frac{p_{2}^{*}}{3}$ & $1-p_{3}^{*}$ & 1/3 & $1-p_{2}^{*}$ &1/6 & $1-p_{n-1}^{*}$ \\
\hline
$S$ & $1-\frac{p_{2}^{*}+p_{n-1}^{*}}{3}$ & $1-p_{n}^{*}$ &$\frac{2-p_{3}^{*}}{3}$ & $1-p_{2}^{*}$ & 1/6 & $1-p_{3}^{*}$ &1/6 & $1-p_{n-1}^{*}$  \\
\hline
\end{tabular}}
\caption{Dynamics of the decrease of the Shapley value.}
\label{dec} 
\vspace{-5mm}
\end{figure}

\textbf{P versus Q:} Since they use the same variable, $\Delta_P= \Delta_Q$ throughout the first segment. 
At the (common) end of the third segment, a simple computation yields $
\Delta_P - \Delta_Q = 0,$ and since $P,Q$ use identical fourth segments, $\Delta_P = \Delta_Q$ throughout their  fourth segment. 

As for the second/third segments, if $p_{n}^{*}<1$ and $p_{n-1}^{*}-p_{n}^{*}>\frac{1}{2}$ then throughout the common portion of the second segment $\Delta_P<\Delta_Q$. Afterwards the difference will start shrinking, and will become positive after a certain value $\lambda_{P,Q}$ 
where $\Delta_{P}=\Delta_{Q}$. Note that at the end of the second segment of $Q$,  
$\Delta_{P}-\Delta_{Q}=\frac{1-p_{n-1}^{*}}{6}\geq 0$, so 
$\lambda_{P,Q}$ is in the second segment of $P$ and the third of $Q$. 

\noindent To determine $\lambda_{P,Q}$ write $\lambda_{P,Q}=1 -p_{2}^{*}+1-p_{n-1}^{*}+\mu_{P,Q}$. We have: 
$\frac{2-p_{n-1}^{*}}{3}(1-p_{n-1}^{*} + \mu_{P,Q}) = (\frac{1}{2}-\frac{p_{n}^{*}}{3})(1-p_{n-1}^{*})+\frac{\mu_{P,Q}}{3}$, $
\mbox{or }\mu_{P,Q}= p_{n-1}^{*} - p_{n}^{*} - \frac{1}{2} 
\mbox{ hence } \lambda_{P,Q}=\frac{3}{2} -p_{2}^{*} - p_{n}^{*}.$

The conclusion is that $\Delta_P\geq \Delta_Q$ for all budgets if $p_{n-1}^{*}-p_{n}^{*}\leq \frac{1}{2}$. 
Otherwise $\Delta_P \geq \Delta_Q$, except for $B\in I_{P,Q}:= (1-p_{2}^{*},\frac{3}{2} -p_{2}^{*} - p_{n}^{*})$. Similar conclusions hold for comparing S versus R. 

\textbf{P versus S:} At the (common) end of their second segment 
$\Delta_P - \Delta_S = (1-p_{n}^{*})(\frac{p_{2}^{*}-1}{3}) + (1-p_{2}^{*})(\frac{1-p_{n}^{*}}{3})=0$. 
So $\Delta_P =  \Delta_S$, and this prevails throughout the third and fourth segments. 

As for the first and second segment, $\Delta_P - \Delta_S \leq 0$ if $p_{3}^{*}+p_{n}^{*}\geq p_{2}^{*}+p_{n-1}^{*}$, $\Delta_P - \Delta_S \geq 0$ if $p_{3}^{*}+p_{n}^{*}\leq p_{2}^{*}+p_{n-1}^{*}$.  Hence $\Delta_P\leq \Delta_S$ for all budgets if $p_{3}^{*}+p_{n}^{*}\geq p_{2}^{*}+p_{n-1}^{*}$. Otherwise $\Delta_P \geq \Delta_S$. 

Summing up: 

-  If $p_{n-1}^{*}-p_{n}^{*}<1/2$, $p_{3}^{*}-p_{2}^{*}<1/2$, $p_{3}^{*}+p_{n}^{*}\leq p_{2}^{*}+p_{n-1}^{*}$, then $\Delta_{P}\geq \Delta_{Q}$, $\Delta_{S}\geq \Delta_{R},\Delta_{P}\geq \Delta_{S}$ for all budgets, so $P$ is optimal. If the last condition is reversed then $S$ is optimal. 

-  If $p_{n-1}^{*}-p_{n}^{*}>1/2$, $p_{3}^{*}-p_{2}^{*}>1/2$ then $\Delta_{P}\leq \Delta_{Q}$ on $I_{P,Q}$, $\Delta_{S}\leq \Delta_{R}$ on $I_{S,R}$. So the best of $Q,R$ is an optimum on $I_{P,Q}\cap I_{S,R}.$ Since $Q,R$ are piecewise linear functions, one of them is better than the other one on an open interval.

\subsection{Proof sketch of Theorems~\ref{kn2} and~\ref{cn2}}

The two proof are very similar, so we only present the one of Theorem~\ref{kn2}.  Particularizing formula~\ref{sh:rel:1} to the case of induced subgraph games, we infer that the Shapley value of player $x$ has the formula $Sh[\overline{v_{FC}}](x)=p_{x}\cdot \sum\limits_{l \in CA(x)} C(x,l) \cdot \frac{2-p_{l}}{2}\mbox{  (*)}
.$


We claim that minimizing $Sh[\overline{v_{FC}}](x)$ is equivalent to solving the following fractional knapsack problem: 
\begin{equation} 
 \left\{\begin{array}{l} 
  max[\sum\limits_{l\in CA(x)} C(x,l)(1-p^{*}_{l})\cdot y_{l}]\\
  \sum\limits_{l\in CA(x)} R_{l}(1-p^{*}_{l})\cdot y_{l} = \sum\limits_{l\in CA(x)} R_{l}\cdot (1-p^{*}_{l}) - B. \\
   0\leq y_{l}\leq 1, \forall l\in CA(x)\\
  \end{array}
\right.
\label{fr-knap}
\end{equation} 
Indeed, by formula~(*) it is only efficient to increase the reliability probabilities of $x$'s authors from $p^{*}_{l}$ to some $p_{l}\in [p^{*}_{l},1]$. If we introduce variables $y_{l}\in [0,1]$ by equation $1-y_{l}=\frac{p_{l}-p^{*}_{l}}{1-p^{*}_{l}}$, (or, equivalently, $ y_{l}=\frac{1-p_{l}}{1-p^{*}_{l}}$), the cost of such move is $R_{l}\cdot (p^{*}_{l}-p_{l})=R_{l}\cdot {(1-y_{l})(1-p^{*}_{l})}$. The total costs must add up to $B$, so $\sum\limits_{l\in CA(x)} R_{l}\cdot {(1-y_{l})(1-p^{*}_{l})} = B$, which is equivalent to system~(\ref{fr-knap}). The minimization of the Shapley value is easily seen to correspond to the maximization of the objective function of~(\ref{fr-knap}). 

Now it is well-known that the greedy algorithm that considers variables $y_{l}$ in decreasing order of their cost/benefit ratio finds an optimal solution to problem~(\ref{fr-knap}). Reinterpreting this result in our language we get the algorithm described in Theorem~\ref{kn2}. 

\section{Related work\protect\footnote{F\lowercase{or reasons of space this section is only sketched.}}}

First of all, \textit{network interdiction} (see e.g. \cite{doi:10.1002/9780470400531.eorms0089,smith2013modern}) is a well-established theme in combinatorial optimization.  Our removal model can be seen as a special case of node interdiction.

Results on the {\it reliability extension} of a cooperative game \cite{meir2012congestion, reliability-games,bachrach2012agent,bachrach2013reliability,bachrach2014cooperative} are naturally related. So is the rich literature on \textit{manipulation}, both in non-cooperative and coalitional settings  \cite{aziz2011false,faliszewski2011multimode,zuckerman2012manipulating,lev2013mergers,Waniek2018, Waniek2017,vallee2014study} and \textit{bribery} \cite{faliszewski2006complexity} in voting. Our framework covers both scenarios, that in which an external perpetrator bribes agents to change their reliabilities, and that in which this is done by a coalition of agents. 

A lot of work has been devoted recently to measuring and characterizing \textit{synergies between players} in multi-agent settings \cite{procaccia2014structure,liemhetcharat2012modeling,liemhetcharat2014weighted}. Synergies between players in cooperative games  are obviously relevant to the theme of this paper: synergic agents' participation to coalitions increases the Shapley value of the given agent. The nature of some of our results (Theorems~\ref{kn},~\ref{kn2} and~\ref{cn2}), that target nodes in a fixed order, provide a concrete way for ranking synergies between these nodes and the attacked one. 

\section{Conclusions and open issues}

Our results have uncovered a rich typology of optimal attacks on players' power indices: Sometimes no attack is beneficial.  Sometimes, the optimal attack is intractable, even when computing the power indices is feasible. For fractional attacks, in many cases (but not always) greedy-type approaches provide an optimal strategy. 

\noindent An open question raised by our work is the complexity of fractional attacks in general full-obligation credit attribution games. Motivated by Theorem~\ref{full-obligation} we believe that even this version is intractable.  On the other hand we would like to see our framework applied to more settings. They include  bicooperative games \cite{bilbao2000bicooperative}, generalized MC-nets \cite{elkind2009tractable}, etc. Of special interest are cases when computing the Shapley value is easy, e.g. voting games with super-increasing weights \cite{bachrach2016analyzing}, flow games on series-parallel networks \cite{elkind2009tractable}, or games with bounded dependency degree \cite{DBLP:conf/aaai/IgarashiIE18}. 

As for relative attacks, we propose studying a more realistic \textit{bicriteria optimization} version of the problem \cite{ravi1993many}: decrease as much as possible the Shapley value of node $x$ while not affecting the Shapley value of node $y$ by more than a certain amount $D$. 

Finally, the related problem of \textit{increasing} the power index of a given node subject to budget constraints is also worth investigating. 


\section*{Acknowledgements}

This work was supported by a grant of Ministry of Research and Innovation, CNCS - UEFISCDI, project number
PN-III-P4-ID-PCE-2016-0842, within PNCDI III.

\bibliographystyle{unsrt} 

\bibliography{Attack_camera-ready_March_3}  

\begin{thebibliography}{10}

\bibitem{reliability-games}
Yoram Bachrach, Reshef Meir, Michal Feldman, and Moshe Tennenholtz.
\newblock Solving cooperative reliability games.
\newblock In {\em Proceedings of the Twenty-Seventh Conference on Uncertainty
  in Artificial Intelligence}, UAI'11, pages 27--34. AUAI Press, 2011.

\bibitem{chalkiadakis2007coalition}
Georgios Chalkiadakis, Evangelos Markakis, and Craig Boutilier.
\newblock Coalition formation under uncertainty: Bargaining equilibria and the
  bayesian core stability concept.
\newblock In {\em Proceedings of the 2007 international joint conference on
  Autonomous agents and multiagent systems}, pages 412--419, 2007.

\bibitem{elkind2012manipulation}
Edith Elkind and G{\'a}bor Erd{\'e}lyi.
\newblock Manipulation under voting rule uncertainty.
\newblock In {\em Proceedings of the 11th International Conference on
  Autonomous Agents and Multiagent Systems-Volume 2}, pages 627--634, 2012.

\bibitem{yokoo2005coalitional}
Makoto Yokoo, Vincent Conitzer, Tuomas Sandholm, Naoki Ohta, and Atsushi
  Iwasaki.
\newblock Coalitional games in open anonymous environments.
\newblock In {\em Proceedings of AAAI}, volume~5, pages 509--514, 2005.

\bibitem{ketchpel1994forming}
Steven Ketchpel.
\newblock Forming coalitions in the face of uncertain rewards.
\newblock In {\em Proceedings of the AAAI}, volume~94, pages 414--419, 1994.

\bibitem{procaccia2014structure}
Ariel~D. Procaccia, Nisarg Shah, and Max~Lee Tucker.
\newblock On the structure of synergies in cooperative games.
\newblock In {\em Proceedings of the AAAI}, pages 763--769, 2014.

\bibitem{barriers-manipulation-voting}
Vincent Conitzer and Toby Walsh.
\newblock {\em Handbook of Computational Social Choice}, chapter Barriers to
  Manipulation in Voting.
\newblock Cambridge University Press, 2016.

\bibitem{control-bribery-voting}
Piotr Faliszewski and Joerg Rothe.
\newblock {\em Handbook of Computational Social Choice}, chapter Control and
  Bribery in Voting.
\newblock Cambridge University Press, 2016.

\bibitem{aziz2011false}
Haris Aziz, Yoram Bachrach, Edith Elkind, and Mike Paterson.
\newblock False-name manipulations in weighted voting games.
\newblock {\em Journal of Artificial Intelligence Research}, 40:57--93, 2011.

\bibitem{elkind2011cloning}
Edith Elkind, Piotr Faliszewski, and Arkadii Slinko.
\newblock Cloning in elections: Finding the possible winners.
\newblock {\em Journal of Artificial Intelligence Research}, 42:529--573, 2011.

\bibitem{zuckerman2012manipulating}
Michael Zuckerman, Piotr Faliszewski, Yoram Bachrach, and Edith Elkind.
\newblock Manipulating the quota in weighted voting games.
\newblock {\em Artificial Intelligence}, 180:1--19, 2012.

\bibitem{lev2013mergers}
Omer Lev, Maria Polukarov, Yoram Bachrach, and Jeffrey~S. Rosenschein.
\newblock Mergers and collusion in all-pay auctions and crowdsourcing contests.
\newblock In {\em Proceedings of the 2013 international conference on
  Autonomous agents and multi-agent systems}, pages 675--682, 2013.

\bibitem{vallee2014study}
Thibaut Vall{\'e}e, Gr{\'e}gory Bonnet, Bruno Zanuttini, and Fran{\c{c}}ois
  Bourdon.
\newblock A study of sybil manipulations in hedonic games.
\newblock In {\em Proceedings of the 2014 international conference on
  Autonomous agents and multi-agent systems}, pages 21--28, 2014.

\bibitem{faliszewski2011multimode}
Piotr Faliszewski, Edith Hemaspaandra, and Lane~A. Hemaspaandra.
\newblock Multimode control attacks on elections.
\newblock {\em Journal of Artificial Intelligence Research}, 40:305--351, 2011.

\bibitem{felsenthal1997ternary}
Dan~S. Felsenthal and Mosh{\'e} Machover.
\newblock Ternary voting games.
\newblock {\em International journal of game theory}, 26(3):335--351, 1997.

\bibitem{bilbao2000bicooperative}
J.M. Bilbao, J.R. Fernandez, A.~Losada~Jim{\'e}nez, and E~Lebr{\'o}n.
\newblock Bicooperative games.
\newblock {\em Cooperative games on combinatorial structures.}, pages 131--295,
  2000.

\bibitem{conitzer2007elections}
Vincent Conitzer, Tuomas Sandholm, and J{\'e}r{\^o}me Lang.
\newblock When are elections with few candidates hard to manipulate?
\newblock {\em Journal of the A.C.M.}, 54(3):14, 2007.

\bibitem{bachrach2014cooperative}
Yoram Bachrach, Rahul Savani, and Nisarg Shah.
\newblock Cooperative max games and agent failures.
\newblock In {\em Proceedings of the 2014 international conference on
  Autonomous agents and multi-agent systems}, pages 29--36, 2014.

\bibitem{deng1994complexity}
Xiaotie Deng and Christos~H. Papadimitriou.
\newblock On the complexity of cooperative solution concepts.
\newblock {\em Mathematics of Operations Research}, 19(2):257--266, 1994.

\bibitem{suri2008determining}
Rama~N. Suri and Y.~Narahari.
\newblock Determining the top-k nodes in social networks using the {S}hapley
  value.
\newblock In {\em Proceedings of the 2008 international conference on
  Autonomous agents and multi-agent systems}, pages 1509--1512, 2008.

\bibitem{aadithya2010game}
Karthik~V. Aadithya and Balaraman Ravindran.
\newblock Game theoretic network centrality: exact formulas and efficient
  algorithms.
\newblock In {\em Proceedings of the 2010 International Conference on
  Autonomous Agents and Multiagent Systems: volume 1-Volume 1}, pages
  1459--1460, 2010.

\bibitem{michalak2013efficient}
Tomasz~P. Michalak, Karthik~V. Aadithya, Piotr~L Szczepanski, Balaraman
  Ravindran, and Nicholas~R. Jennings.
\newblock Efficient computation of the shapley value for game-theoretic network
  centrality.
\newblock {\em Journal of Artificial Intelligence Research}, pages 607--650,
  2013.

\bibitem{tarkowski2017game}
Mateusz~K. Tarkowski, Tomasz~P. Michalak, Talal Rahwan, and Michael Wooldridge.
\newblock Game-theoretic network centrality: A review.
\newblock {\em arXiv preprint arXiv:1801.00218}, 2017.

\bibitem{papapetrou2011shapley}
Panagiotis Papapetrou, Aristides Gionis, and Heikki Mannila.
\newblock A shapley value approach for influence attribution.
\newblock In {\em Machine Learning and Knowledge Discovery in Databases}, pages
  549--564. Springer, 2011.

\bibitem{karpov2014equal}
Alexander Karpov.
\newblock Equal weights coauthorship sharing and the shapley value are
  equivalent.
\newblock {\em Journal of Informetrics}, 8(1):71--76, 2014.

\bibitem{elkind2009tractable}
Edith Elkind, Leslie~Ann Goldberg, Paul~W. Goldberg, and Michael Wooldridge.
\newblock A tractable and expressive class of marginal contribution nets and
  its applications.
\newblock {\em Mathematical Logic Quarterly}, 55(4):362--376, 2009.

\bibitem{ieong2005marginal}
Samuel Ieong and Yoav Shoham.
\newblock Marginal contribution nets: a compact representation scheme for
  coalitional games.
\newblock In {\em Proceedings of the 6th ACM conference on Electronic
  commerce}, pages 193--202. ACM, 2005.

\bibitem{bachrach2008coalitional}
Yoram Bachrach and Jeffrey~S Rosenschein.
\newblock Coalitional skill games.
\newblock In {\em Proceedings of the 2007 international joint conference on
  Autonomous agents and multiagent systems-Volume 2}, pages 1023--1030, 2008.

\bibitem{bachrach2013computing}
Yoram Bachrach, David Parkes, and Jeffrey~S. Rosenschein.
\newblock Computing cooperative solution concepts in coalitional skill games.
\newblock {\em Artificial Intelligence}, 204:1--21, 2013.

\bibitem{matsui2001np}
Yasuko Matsui and Tomomi Matsui.
\newblock Np-completeness for calculating power indices of weighted majority
  games.
\newblock {\em Theoretical Computer Science}, 263(1-2):305--310, 2001.

\bibitem{aziz2009power}
Haris Aziz, Oded Lachish, Mike Paterson, and Rahul Savani.
\newblock Power indices in spanning connectivity games.
\newblock In {\em International Conference on Algorithmic Applications in
  Management}, pages 55--67. Springer, 2009.

\bibitem{michalak2015defeating}
Tomasz~P. Michalak, Talal Rahwan, Oskar Skibski, and Michael Wooldridge.
\newblock Defeating terrorist networks with game theory.
\newblock {\em IEEE Intelligent Systems}, 30(1):53--61, 2015.

\bibitem{lindelauf2013cooperative}
R.~Lindelauf, H.~Hamers, and B.~Husslage.
\newblock Cooperative game theoretic centrality analysis of terrorist networks:
  The cases of jemaah islamiyah and al qaeda.
\newblock {\em European Journal of Operational Research}, 229(1):230--238,
  2013.

\bibitem{Morgan2018}
Allison~C. Morgan, Dimitrios~J. Economou, Samuel~F. Way, and Aaron Clauset.
\newblock Prestige drives epistemic inequality in the diffusion of scientific
  ideas.
\newblock {\em EPJ Data Science}, 7(1):40, Oct 2018.

\bibitem{hirsch2018mathbf}
Jorge~E. Hirsch.
\newblock $h_{\alpha}$ : An index to quantify an individual's scientific
  leadership.
\newblock {\em arXiv preprint arXiv:1810.01605}, 2018.

\bibitem{chalkiadakis2011computational}
George Chalkiadakis, Edith Elkind, and Michael Wooldridge.
\newblock Computational aspects of cooperative game theory.
\newblock {\em Synthesis Lectures on Artificial Intelligence and Machine
  Learning}, 2011.

\bibitem{doi:10.1002/9780470400531.eorms0089}
J.~Cole Smith.
\newblock Basic interdiction models.
\newblock In {\em Wiley Encyclopedia of Operations Research and Management
  Science}. American Cancer Society, 1997.

\bibitem{smith2013modern}
J.~Cole Smith, Mike Prince, and Joseph Geunes.
\newblock Modern network interdiction problems and algorithms.
\newblock In {\em Handbook of combinatorial optimization}, pages 1949--1987.
  Springer, 2013.

\bibitem{meir2012congestion}
Reshef Meir, Moshe Tennenholtz, Yoram Bachrach, and Peter Key.
\newblock Congestion games with agent failures.
\newblock In {\em Proceedings of the AAAI}, volume~12, pages 1401--1407, 2012.

\bibitem{bachrach2012agent}
Yoram Bachrach, Ian Kash, and Nisarg Shah.
\newblock Agent failures in totally balanced games and convex games.
\newblock In {\em International Workshop on Internet and Network Economics},
  pages 15--29. Springer, 2012.

\bibitem{bachrach2013reliability}
Yoram Bachrach and Nisarg Shah.
\newblock Reliability weighted voting games.
\newblock In {\em International Symposium on Algorithmic Game Theory}, pages
  38--49. Springer, 2013.

\bibitem{Waniek2018}
Marcin Waniek, Tomasz~P. Michalak, Michael~J. Wooldridge, and Talal Rahwan.
\newblock Hiding individuals and communities in a social network.
\newblock {\em Nature Human Behaviour}, 2(2):139--147, 2018.

\bibitem{Waniek2017}
Marcin Waniek, Tomasz~P. Michalak, Talal Rahwan, and Michael Wooldridge.
\newblock On the construction of covert networks.
\newblock In {\em Proceedings of the 16th Conference on Autonomous Agents and
  MultiAgent Systems}, AAMAS '17, pages 1341--1349, Richland, SC, 2017.
  International Foundation for Autonomous Agents and Multiagent Systems.

\bibitem{faliszewski2006complexity}
Piotr Faliszewski, Edith Hemaspaandra, and Lane~A. Hemaspaandra.
\newblock The complexity of bribery in elections.
\newblock In {\em Proceedins of the AAAI}, volume~6, pages 641--646, 2006.

\bibitem{liemhetcharat2012modeling}
Somchaya Liemhetcharat and Manuela Veloso.
\newblock Modeling and learning synergy for team formation with heterogeneous
  agents.
\newblock In {\em Proceedings of the 2012 International Conference on
  Autonomous Agents and Multiagent Systems-Volume 1}, pages 365--374, 2012.

\bibitem{liemhetcharat2014weighted}
Somchaya Liemhetcharat and Manuela Veloso.
\newblock Weighted synergy graphs for effective team formation with
  heterogeneous ad hoc agents.
\newblock {\em Artificial Intelligence}, 208:41--65, 2014.

\bibitem{bachrach2016analyzing}
Yoram Bachrach, Yuval Filmus, Joel Oren, and Yair Zick.
\newblock Analyzing power in weighted voting games with super-increasing
  weights.
\newblock In {\em International Symposium on Algorithmic Game Theory}, pages
  169--181. Springer, 2016.

\bibitem{DBLP:conf/aaai/IgarashiIE18}
Ayumi Igarashi, Rani Izsak, and Edith Elkind.
\newblock Cooperative games with bounded dependency degree.
\newblock In {\em Proceedings of the Thirty-Second {AAAI} Conference on
  Artificial Intelligence, (AAAI-18), the 30th innovative Applications of
  Artificial Intelligence (IAAI-18), and the 8th {AAAI} Symposium on
  Educational Advances in Artificial Intelligence (EAAI-18), New Orleans,
  Louisiana, USA, February 2-7, 2018}, pages 1063--1070, 2018.

\bibitem{ravi1993many}
R.~Ravi, Madhav~V. Marathe, S.S. Ravi, Daniel~J. Rosenkrantz, and Harry~B.
  Hunt~III.
\newblock Many birds with one stone: Multi-objective approximation algorithms.
\newblock In {\em Proceedings of the twenty-fifth annual ACM Symposium on
  Theory of Computing}, pages 438--447, 1993.

\end{thebibliography}

\newpage
\section*{Appendix}

\section{Proof of Corollary~\ref{centrality}}

We will actually give the following formulas for the partial derivatives of the Shapley value: 

\begin{align*} 
\frac{\partial Sh[\overline{v_{NC_1}}](x)}{\partial p_{j}} & = -p_{x}\mathlarger{\sum}\limits_{\stackrel{y\in \widehat{N(
x)}\cap \widehat{N(j)}}{S\subseteq \widehat{N(y)}\setminus \{x,j\}}} \frac{1}{(|S|+1)(|S| + 2)} \Pi_{S, \widehat{N(y)}\setminus \{x,j\}}< 0. 
\end{align*}
\begin{align*} 
& \frac{\partial Sh[\overline{v_{NC_2}}](x)}{\partial p_{j}}  = -p_{x}\mathlarger{\sum}\limits_{S\subseteq N(x)\setminus \{j\}} \frac{k}{|(|S|+1)(|S| + 2)} \Pi_{S, N(x)\setminus \{j\}} - p_{x}\cdot \\
& \mathlarger{\sum}\limits_{y\in N(
x)\cap \widehat{N(j)}}\mathlarger{\sum}\limits_{\stackrel{S\subseteq \widehat{N(y)}\setminus \{x,j\}}{|S|\geq k-1}} [\frac{|S|+2-k}{(|S|+1)(|S| + 2)}-\frac{|S|+1-k}{|S|(|S|+1)}]\Pi_{S, \widehat{N(y)}\setminus \{x,j\}}  \\ 
& = -p_{x}\mathlarger{\sum}\limits_{S\subseteq N(x)\setminus \{j\}} \frac{k}{|(|S|+1)(|S| + 2)} \Pi_{S, N(x)\setminus \{j\}}-p_{x}\cdot \\
& \mathlarger{\sum}\limits_{y\in N(
x)\cap \widehat{N(j)}}\mathlarger{\sum}\limits_{\stackrel{S\subseteq \widehat{N(y)}\setminus \{x,j\}}{|S|\geq k-1}}[ \frac{k}{(|S|+1)(|S|+2)}+\frac{k-1}{|S|(|S|+1)}]\Pi_{S, \widehat{N(y)}\setminus \{x,j\}}
\end{align*}
which is less than zero. 
\begin{align*} 
 \frac{\partial Sh[\overline{v_{NC_3}}](x)}{\partial p_{j}} & = -p_{x}\mathlarger{\sum}\limits_{\stackrel{y\in \widehat{N(x)}}{S\subseteq \widehat{N(y)}\setminus \{x,j\}}} \frac{1}{|(|S|+1)(|S| + 2)} \Pi_{S, \widehat{N_{cut}(y)}\setminus \{x,j\}} 
\end{align*}
is less than zero 
(if $j\not \in B(1,2)$ all derivatives are zero). 

These formulas follow easily from Theorem~\ref{sh-gamma1}, by applying the linearity of partial derivatives and Claim~\ref{aux2}. 

\section{Proof of Lemma~\ref{sign}}

First, for $G=K_n$, $j\neq l$ distinct from $x=1$ and $W=V\setminus \{1,j,l\}$, the formula from the proof of Corollary~\ref{centrality} particularizes to: 
\begin{align*} 
& \frac{\partial Sh[\overline{v_{NC_1}}](1)}{\partial p_{j}}  = -p_{1}N \mathlarger{\sum}\limits_{S\subseteq V\setminus \{1,j\}} \frac{1}{|(|S|+1)(|S| + 2)} \Pi_{S, V\setminus \{1,j\}} = - p_{1}N \cdot \\
&   [ \mathlarger{\sum}\limits_{\stackrel{S\subseteq V\setminus \{1,j\}} {l\in S}}\frac{ \Pi_{S, V\setminus \{1,j\}} }{|(|S|+1)(|S| + 2)} + \mathlarger{\sum}\limits_{\stackrel{S\subseteq V\setminus \{1,j\}} {l\not \in S}} \frac{ \Pi_{S, V\setminus \{1,j\}} }{|(|S|+1)(|S| + 2)}]=\\
& =  - p_{1}N [ \mathlarger{\sum}\limits_{S\subseteq W}\frac{p_{l}}{|(|S|+2)(|S| + 3)} \Pi_{S, W}  + \mathlarger{\sum}\limits_{S\subseteq W} \frac{1-p_{l}}{|(|S|+1)(|S| + 2)} \Pi_{S, W}]= \\
& = - p_{1}N \mathlarger{\sum}\limits_{S\subseteq V\setminus \{x,j,l\}}\frac{p_{l}(|S|+1)+(1-p_{l})(|S|+3)}{(|S|+1)(|S|+2)(|S| + 3)} \Pi_{S, W}=  \\ 
& = - p_{1}N  \mathlarger{\sum}\limits_{S\subseteq W}\frac{(|S|+3)-2p_{l}}{(|S|+1)(|S|+2)(|S| + 3)} \Pi_{S, W}. 
\end{align*}

For $j\neq l$

\begin{align*} 
& \frac{\partial Sh[\overline{v_{NC_1}}](1)}{\partial p_{j}} - \frac{\partial Sh[\overline{v_{NC_1}}](1)}{\partial p_{l}} =  -p_{1}N  \mathlarger{\sum}\limits_{S\subseteq W}\frac{2(p_{l}-p_{j})\Pi_{S, W}}{(|S|+1)(|S|+2)(|S| + 3)}  \\
& = 2 p_{1}N (p_{j}-p_{l}) \mathlarger{\sum}\limits_{S\subseteq W} \frac{1}{(|S|+1)(|S|+2)(|S| + 3)} \Pi_{S, W}. \\
\end{align*}

The sum is composed of nonnegative terms, so the difference of partial derivatives has the same sign as the difference $p_{j}-p_{l}$. 

For $G=S_{n}$, with $1$ the center of the star, $\widehat{N(1)}\cap \widehat{N(j)} = \{1,j\}$, so the particularization of the formula reads
\begin{align*} 
& \frac{\partial Sh[\overline{v_{NC_1}}](1)}{\partial p_{j}} = -p_{1}\mathlarger{\sum}\limits_{\stackrel{y\in \{1,j\}}{S\subseteq \widehat{N(y)}\setminus \{1,j\}}} \frac{1}{|(|S|+1)(|S| + 2)} \Pi_{S, \widehat{N(y)}\setminus \{1,j\}} = \\
&  -p_{1}[\frac{1}{2}\Pi_{\emptyset,\emptyset}+ \mathlarger{\sum}\limits_{S\subseteq V\setminus \{1,j\}} \frac{1}{|(|S|+1)(|S| + 2)} \Pi_{S, V\setminus \{1,j\}}] = -p_{1}[\frac{1}{2}+ \\
& \mathlarger{\sum}\limits_{\stackrel{S\subseteq V\setminus \{1,j\}}{l\in S}} \frac{1}{|(|S|+1)(|S| + 2)} \Pi_{S, V\setminus \{1,j\}}+ \mathlarger{\sum}\limits_{\stackrel{S\subseteq V\setminus \{1,j\}}{l\not \in S}} \frac{1}{|(|S|+1)(|S| + 2)} \cdot \\
& \cdot \Pi_{S, V\setminus \{1,j\}} ] = -p_{1}[\frac{1}{2}+  \mathlarger{\sum}\limits_{S\subseteq W} \frac{p_{l}}{|(|S|+2)(|S| + 3)} \Pi_{S, W}+ \\ 
& + \mathlarger{\sum}\limits_{S\subseteq W} \frac{1-p_{l}}{|(|S|+1)(|S| + 2)} \Pi_{S, W} = -p_{1}[\frac{1}{2}+ \\ 
&   + \mathlarger{\sum}\limits_{S\subseteq W} \frac{p_{l}(|S|+1)+(1-p_{l})(|S|+3)}{(|S|+1)(|S|+2)(|S| + 3)} \Pi_{S, W}= 
-p_{1}[\frac{1}{2}+\\ 
& + \mathlarger{\sum}\limits_{S\subseteq W} \frac{(|S|+3)-2p_{l}}{(|S|+1)(|S|+2)(|S| + 3)} \Pi_{S, W}
\end{align*} 
So again, for $j\neq l$, the sign of the difference of partial derivatives is the same as the sign of the term $p_{j}-p_{l}$, as 

\begin{align*} 
& \frac{\partial Sh[\overline{v_{NC_1}}](1)}{\partial p_{j}} - \frac{\partial Sh[\overline{v_{NC_1}}](1)}{\partial p_{l}} =  -p_{1} \mathlarger{\sum}\limits_{S\subseteq W}\frac{2(p_{l}-p_{j})}{(|S|+1)(|S|+2)(|S| + 3)} \Pi_{S, W}. \\
& = 2 p_{1} (p_{j}-p_{l}) \mathlarger{\sum}\limits_{S\subseteq W} \frac{1}{(|S|+1)(|S|+2)(|S| + 3)} \Pi_{S, W}. \\
\end{align*}

\section{Proof of Lemma~\ref{equal}}

If $G=K_n$ then 
\[
Sh[\overline{v_{NC_{1}}}](1)\rvert^{p(\epsilon)}_{p} = 
p_{1}N \sum_{S\subseteq V\setminus 1} \frac{\Pi_{S,V\setminus 1}\rvert^{p(\epsilon)}_{p}}{|S|+1}
\]
We denote $W=V\setminus \{1,i,j\}$. Expanding terms corresponding to agents $i,j$ (who may or may not be live) and grouping we infer that 
\begin{align*}
& Sh[\overline{v}](1)|^{p(\epsilon)}_{p} = 
p_{1}N \sum_{T\subseteq W} \Pi_{T,W}\cdot [\frac{(p_{i}+\epsilon)(p_{j}-\epsilon)-p_{i}p_{j}}{|S|+3} +\\ 
& \frac{(p_{i}+\epsilon)(1-p_{j}+\epsilon)-p_{i}(1-p_{j})}{|S|+2}+\frac{(1-p_{i}-\epsilon)(p_{j}-\epsilon)-(1-p_{i})p_{j}}{|S|+2}\\ & +\frac{(1-p_{i}-\epsilon)(1-p_{j}+ \epsilon)-(1-p_{i})(1-p_{j})}{|S|+1}= p_{x}N \sum_{T\subseteq W} \Pi_{T,W}\cdot \\ 
& [\frac{\epsilon(p_{j}-p_{i})-\epsilon^{2}}{|S|+3}+\frac{\epsilon(1-p_{j}+p_{i}))+\epsilon^{2}}{|S|+2}+\frac{-\epsilon(1+p_{j}-p_{i})-\epsilon^{2}}{|S|+2}+\\
& + \frac{\epsilon(p_{j}-p_{i})-\epsilon^{2}}{|S|+1}]= 
 p_{1}N \sum_{T\subseteq W} \Pi_{T,W}\cdot [\frac{-\epsilon^{2}}{|S|+3}+\frac{-\epsilon^2}{|S|+1}]=  \\
& = - p_{1}N\epsilon^2 \sum_{T\subseteq W} \Pi_{T,W}\cdot [\frac{1}{|S|+1}-\frac{1}{|S|+3}]<0 \mbox{ if }\epsilon \neq 0. 
\end{align*}
Similarly, if $G=S_n$ then 
\begin{align*}
& Sh[\overline{v}](1)|^{p(\epsilon)}_{p}= 
p_{1}[\sum_{k=2}^{n} \sum_{S\subseteq \{k\}} \frac{ \Pi_{S,\{k\}}|^{p(\epsilon)}_{p}}{|S|+1}+ \sum_{S\subseteq V\setminus \{1\}} \frac{\Pi_{S,V\setminus 1}|^{p(\epsilon)}_{p}}{|S|+1}]\\ 
& = p_{1}[\sum_{k=2}^{n} (1-\frac{p_{k}}{2})|^{p(\epsilon)}_{p} + \sum_{S\subseteq V\setminus \{1\}} \frac{\Pi_{S,V\setminus 1}|^{p(\epsilon)}_{p}}{|S|+1}] \\ & = p_{1}\sum_{S\subseteq V\setminus \{1\}} \frac{\Pi_{S,V\setminus 1}|^{p(\epsilon)}_{p}}{|S|+1}
\end{align*}
In the last calculation we took into account the fact that the first sum has two nonzero difference terms, one with value $\epsilon/2$ and one with value $-\epsilon/2$ (corresponding to $k=i,j$, respectively) 
which cancel each other. What remains is (up to a multiplicative factor of N) identical to the difference in the case $G=K_{n}$, and the rest of the proof is identical. 

\section{Proof of Theorem~\ref{kn}, the case $G=S_{n}$, with vertex 1 not being central}

For $G=S_{n}$, with $2$ the center of the star, $\widehat{N(1)}\cap \widehat{N(j)} = \{2\}$, if $j\neq 2$, $\widehat{N(1)}\cap \widehat{N(2)} = \{1,2\}$.  Therefore, for $j\neq 1,2$: 
\begin{align*} 
& \frac{\partial Sh[\overline{v_{NC_1}}](1)}{\partial p_{j}} = -p_{1}\mathlarger{\sum}\limits_{S\subseteq \widehat{N(2)}\setminus \{1,j\}} \frac{1}{|(|S|+1)(|S| + 2)} \Pi_{S, \widehat{N(y)}\setminus \{1,j\}} = \\
&  -p_{1} \mathlarger{\sum}\limits_{S\subseteq V\setminus \{1,j\}} \frac{1}{|(|S|+1)(|S| + 2)} \Pi_{S, V\setminus \{1,j\}} 
\end{align*}
the same, within a factor of $N$, as the partial derivative for the case $G=K_{n}$, while 
\begin{align*} 
& \frac{\partial Sh[\overline{v_{NC_1}}](1)}{\partial p_{2}} = -p_{1}\mathlarger{\sum}\limits_{\stackrel{y\in \{1,2\}}{S\subseteq \widehat{N(y)}\setminus \{1,2\}}} \frac{1}{|(|S|+1)(|S| + 2)} \Pi_{S, \widehat{N(y)}\setminus \{1,2\}} = \\
&  -p_{1}[\frac{1}{2}+ \mathlarger{\sum}\limits_{S\subseteq V\setminus \{1,2\}} \frac{1}{|(|S|+1)(|S| + 2)} \Pi_{S, V\setminus \{1,2\}}] 
\end{align*}
SInce $
\mathlarger{\sum}\limits_{S\subseteq V\setminus \{1,j\}} \Pi_{S, V\setminus \{1,j\}} = 1$, 
$\frac{1}{|(|S|+1)(|S| + 2)}\leq \frac{1}{2}$, strictly less if $S\neq \emptyset$, and $\Pi_{V\setminus \{1,2\},V\setminus \{1,2\}}>0$, we infer that 
\[
 \frac{\partial Sh[\overline{v_{NC_1}}](1)}{\partial p_{j}} > -\frac{1}{2}, \forall j\neq 1,2
\]
while 
\[
 \frac{\partial Sh[\overline{v_{NC_1}}](1)}{\partial p_{2}} \leq -\frac{1}{2}
\]
Hence, by an application of Lemma~\ref{aux}, in the optimal solution $p_{j}>p_{j}^{*}\Rightarrow p_{2}=1$, in other words 
reliability of node 2 must be increased (to one, if the budget allows it), before any other reliability is increased. 

Invoking the result for $G=K_{n}$ (Lemma~\ref{sign}), we infer that 
\[
sign\big( \frac{\partial Sh[\overline{v_{NC_1}}](1)}{\partial p_{j}}|_{p} -  \frac{\partial Sh[\overline{v_{NC_1}}](1)}{\partial p_{j}}|_{p} \big)=sign(p_{j}-p_{l})
\]

Given this result, the rest of the proof is the same as in the cases $G=K_n$, $G=S_{n}$ with 1 in the center. 

\section{Proof of Theorem~\ref{full-obligation}}

\begin{definition}
The Budgeted Max-Coverage problem is specified as follows: 
\begin{itemize} 
\item[-] GIVEN: Universe $U=\{u_{1},u_{2},\ldots, u_{n}\}$, each $u_{i}$ coming with a positive integer \textit{weight} $w_{i}$, subsets $P_{1},P_{2},\ldots, P_{m}$ of $U$, each set $P_{i}$ coming with a positive integer \textit{cost} $c_{i}$, and integers $k,L\geq 1$. 
\item[-] TO DECIDE: Can we choose some sets such that 
\begin{itemize} 
\item[-] the total cost of the chosen sets is at most $k$. 
\item[-] the sum of weights of elements covered by the chosen sets is at least $L$ ? 
\end{itemize} 
\end{itemize} 
\end{definition} 

Budgeted Max-Coverage is not only NP-complete, but even hard to approximate \cite{khuller1999budgeted}. 

We reduce an instance of Budgeted Max Coverage to the problem of minimizing the Shapley value of player 1 under the full obligation model as follows: 
\begin{itemize} 
\item[-] all players will have baseline reliabilities equal to one. 
\item[-] elements will correspond to "papers" coauthored by 1 and some other players.  
\item[-] sets will correspond to coauthors of $1$, representing, for each coauthor, the paper he coauthored together with 1. We assume that all papers of 1 are written in collaboration. 
\item[-] The cost of bringing (in a removal attack) the reliability probability of a given player to zero is the cost of the associated set in the instance of Weighted Max Coverage. 
\item[-] the score of a paper is equal to the number of authors times the weight of the associated element. This way, author 1 gets a contribution from a paper in its Shapley value equal to the weight of the associated elements. 
\end{itemize}

Targeting a set of players of total cost at most $k$ will reduce the Shapley value of player 1 by precisely the total weight of elements covered by these players. Thus one can reduce the Shapley value of player 1 by at least $L$ iff the answer to the corresponding Budgeted Max Coverage problem is positive.

\end{document}